\documentclass[letterpaper,11pt]{article}
\usepackage[sort,compress]{cite}
\usepackage[top=1.1in, bottom=1.13in, left=1.1in, right=1.1in]{geometry}
\usepackage[utf8]{inputenc}
\usepackage{amsmath}
\usepackage{amssymb}
\usepackage{amsfonts}
\usepackage{amsthm}
\usepackage{graphicx}
\usepackage{tikz}
\usepackage{xspace}
\usepackage{subcaption}
\usepackage{color}
\usepackage{comment}
\usepackage{titling}
\usepackage{enumitem}
\usepackage[font={small}]{caption}
\usepackage{changepage}
\usepackage[export]{adjustbox}
\usepackage{mathtools}
\usepackage{rotating}
\usepackage{float}
\usepackage[normalem]{ulem}
\usepackage[textsize=tiny]{todonotes}

\usepackage[skip=0pt]{caption} 

\usepackage{hyperref}

\makeatletter
\newcommand\footnoteref[1]{\protected@xdef\@thefnmark{\ref{#1}}\@footnotemark}
\makeatother

\setlist[itemize]{leftmargin=1em,topsep=0.2em,noitemsep}
\setlist[enumerate]{leftmargin=1.4em,topsep=0.2em,noitemsep,labelsep=0.3em}

\usetikzlibrary{decorations.pathmorphing}
\usetikzlibrary{shapes.misc} 
\usetikzlibrary{backgrounds}

\newtheorem{theorem}{Theorem}[section]

\newtheorem{lemma}[theorem]{Lemma}
\newtheorem{corollary}[theorem]{Corollary}

\newtheorem{observation}[theorem]{Observation}
\theoremstyle{definition}
\newtheorem{definition}[theorem]{Definition}
\newtheorem{fact}[theorem]{Fact}

\newcommand{\termasm}[1]{\mathcal{A}_{\Box}[{#1}]}
\newcommand{\pathassembly}[1]{\mathrm{asm}{(#1)}}
\newcommand{\asm}[1]{\pathassembly{#1}}

\newcommand{\Z}{\mathbb{Z}}
\newcommand{\N}{\mathbb{N}}

\newcommand{\resp}{respectively\xspace}
\newcommand{\vect}{\protect\overrightarrow}
\newcommand{\dom}[1]{{\rm dom}(#1)}
\newcommand{\prodasm}[1]{\mathcal{A}[{#1}]}
\newcommand{\calT}{\mathcal{T}}
\newcommand{\prodT}{\prodasm{\mathcal{T}}}
\newcommand{\prodpaths}[1]{{\bf{P}}[{#1}]}
\newcommand{\prodpathsT}{\prodpaths{\mathcal{T}}}
\newcommand{\pos}[1]{\mathrm{pos}(#1)}
\newcommand\type[1]{\mathrm{type}(#1)}

\newcommand{\para}{k}
\newcommand{\parb}{n}
\newcommand{\series}{h}
\newcommand{\step}{i} 
\newcommand{\inds}{i'} 

\newcommand{\settiles}{\mathcal{T}}
\newcommand{\glueset}{\mathcal{T}^{(\para,\parb)}}

\newcommand{\glueg}{\mathsf{g}}
\newcommand{\gluey}{\mathsf{o}}
\newcommand{\glueb}{\mathsf{b}}
\newcommand{\gluep}{\mathsf{r}}

\newcommand{\tiles}{s}
\newcommand{\tileg}{g}
\newcommand{\tiley}{o}
\newcommand{\tileb}{b}
\newcommand{\tilep}{r}

\newcommand{\gluesets}{\mathcal{S}}
\newcommand{\gluesetg}{\mathcal{G}}
\newcommand{\gluesety}{\mathcal{O}}
\newcommand{\gluesetb}{\mathcal{B}}
\newcommand{\gluesetp}{\mathcal{R}}

\newcommand{\tg}{G}
\newcommand{\ty}{O}
\newcommand{\tb}{B}
\newcommand{\tp}{R}

\newcommand{\patha}{P}
\newcommand{\pathaa}{P'}
\newcommand{\dead}[1]{D^{#1}}
\newcommand{\deadd}[2]{D'^{(#1,#2)}}
\newcommand{\pathg}{G}
\newcommand{\pathy}[1]{O^{#1}}
\newcommand{\pathb}[1]{B^{#1}}
\newcommand{\pathpurple}{R}
\newcommand{\pathp}[1]{R^{(#1,#1)}}
\newcommand{\pathpp}[2]{R^{(#2,#1)}}
\newcommand{\paths}[2]{S^{(#1,#2)}}
\newcommand{\pathl}{L}

\newcommand{\vu}[2]{\overrightarrow{v}^{(#1,#2)}}

\newcommand{\tas}{T^{(\para,\parb)}} 
\newcommand{\tasr}[1]{T^{(#1,\parb)}}

\definecolor{xgreen}{RGB}{0,128,0}
\definecolor{xblue}{RGB}{0,204,255}
\definecolor{xorange}{RGB}{255,165,0}
\definecolor{xred}{RGB}{255,0,0}

\newcommand{\draws}[2]{\draw[fill=white] (#1.16,#2.16) rectangle (#1.84,#2.84);}
\newcommand{\drawg}[2]{\draw[fill=xgreen!50!white] (#1.16,#2.16) rectangle (#1.84,#2.84);}
\newcommand{\drawo}[2]{\draw[fill=xorange!50!white] (#1.16,#2.16) rectangle (#1.84,#2.84);}
\newcommand{\drawb}[2]{\draw[fill=xblue!50!white] (#1.16,#2.16) rectangle (#1.84,#2.84);}
\newcommand{\drawr}[2]{\draw[fill=xred!50!white] (#1.16,#2.16) rectangle (#1.84,#2.84);}

\title{Non-cooperatively assembling large structures: \\a $2D$ pumping lemma cannot be as powerful as its $1D$ counterpart.
\thanks{
Research supported by European Research Council (ERC) under the European Union's Horizon 2020 
research and innovation programme (grant agreement No 772766, Active-DNA project), and Science 
Foundation Ireland (SFI) under Grant number 18/ERCS/5746.
}
}

\author{Pierre-\'Etienne Meunier\\
  Inria  \\
\href{mailto:pierre-etienne.meunier@inria.fr}{pierre-etienne.meunier@inria.fr}
\and
Damien Regnault\\
IBISC, Univ Évry, Université Paris-Saclay,\\
91025, Evry, France.\\
\href{mailto:damien.regnault@univ-evry.fr}{damien.regnault@univ-evry.fr}
}
\date{}

\begin{document}
\numberwithin{figure}{section}
\maketitle

\begin{abstract}
We show the first asymptotically efficient constructions in the so-called "noncooperative planar tile assembly" model.

Algorithmic self-assembly is the study of the local, distributed, asynchronous algorithms ran by molecules to self-organise, in particular during crystal growth. The general \emph{cooperative} model, also called \emph{"temperature 2"}, uses \emph{synchronisation} to simulate Turing machines, build shapes using the smallest possible amount of tile types, and other algorithmic tasks.
However, in the non-cooperative ("temperature 1") model, the growth process is entirely asynchronous, and mostly relies on geometry. Even though the model looks like a generalisation of finite automata to two dimensions, its 3D generalisation is capable of performing arbitrary (Turing) computation [SODA 2011], and of universal simulations [SODA 2014], whereby a single 3D non-cooperative tileset can simulate the dynamics of all possible 3D non-cooperative systems, up to a constant scaling factor.

However, we showed in [STOC 2017] that the original 2D non-cooperative model is not capable of universal simulations, and the question of its computational power is still widely open. Here, we show an unexpected result, namely that this model can reliably grow assemblies of size $\Omega(n \log n)$ with only $n$ tile types, which is the first asymptotically efficient positive construction.

\end{abstract}

\clearpage
\section{Introduction}

Our ability to understand and control matter at the scale of molecules conjures a future where we can engineer our own materials, interact with biological networks to cure their malfunctions, and build molecular computers and nanoscale factories.
The field of molecular computing and molecular self-assembly studies the algorithms run by molecules to exchange information, to self-organise in space, to grow and replicate. Our goal is to build a theory of how they compute, and of how we can program them.

One of the most successful models of algorithmic self assembly is the \emph{abstract tile assembly model}, imagined by Winfree~\cite{Winf98}. In that model, we start with a single seed assembly and a finite number of tile types (with an infinite supply of each type), and attach tiles, one at a time, asynchronously and nondeterministically, to the assembly, based on a condition on their borders' colours.

This model has served to bootstrap the field of molecular computing, which has since produced an impressive number of experimental realisations, from DNA motors~\cite{yurke2000dna} to arbitrary hard-coded shapes at the nanoscale~\cite{RothOrigami}, and cargo-sorting robots~\cite{thubagere2017cargo}.

\newcommand\turinginothermodels{\cite{RotWin00,SolWin07,Versus,OneTile,Cook-2011,Roth01,Patitz-2011,Signals,Hendricks-2014,Fekete2014,gilbert2015continuous,Winf98,Winfree98simulationsof}}

On the theoretical side, the abstract tile assembly model has been used to explore different features of self-organisation in space, especially in an asynchronous fashion. In many variants of the model, tile assembly is capable of simulating Turing machines~\turinginothermodels. More surprisingly, in its original form, the model is \emph{intrinsically universal}~\cite{IUSA}, meaning that there is a single ``universal'' tileset capable of simulating the behaviour of any other tileset (that behaviour is encoded only in the seed of the universal tileset).

In the usual form of the model, a part of the assembly can ``wait'' for another to grow long enough to cooperate. In the non-cooperative model, however, any tile can attach to any location, as long as at least one side matches the colour of that location. Therefore, ``synchronising'' different parts of the assembly is impossible, and the main question becomes, \emph{what kind of computation can we do in a completely asynchronous way?} The answer seems to depend crucially on the space in which the assemblies grow: in one dimension, non-cooperative tile assembly is equivalent to finite automata\footnote{Actually, deterministic tile assembly systems map directly to deterministic finite automata.}, and are therefore not too powerful. In three dimensions though, this model is capable of simulating Turing machines~\cite{Cook-2011}, and even of simulating itself \emph{intrinsically}~\cite{Meunier-2014}. If instead of square tiles, we use tiles that do \emph{not} tile the plane, the situation becomes even more puzzling: tiles whose shape are regular polygons can perform arbitrary computation, but only if they have at least seven sides~\cite{DBLP:conf/soda/GilbertHPR16}. In a similar way, polyomino tiles can also simulate Turing machines, provided that at least one of their dimensions is at least two~\cite{DBLP:conf/soda/FeketeHPRS15}.

However, in two dimensions with regular square tiles, the capabilities of this model remain largely mysterious. All we know is that it cannot simulate the general (cooperative) model up to rescaling~\cite{Meunier-2014}, and cannot simulate itself either~\cite{STOC2017}, but we know very little about its actual computational power. A number of related questions and simpler models have been studied to try and approach this model: a probabilistic assembly schedule~\cite{Cook-2011}, negative glues~\cite{SingleNegative}, no mismatches~\cite{Doty-2011,Manuch-2010}, and different tile shapes~\cite{DBLP:conf/soda/GilbertHPR16,DBLP:conf/soda/FeketeHPRS15}.

Due to the proximity with finite automata, a first intuition is that we can try to ``pump'' parts of an assembly between two tiles of equal type, resulting in infinite, periodic paths, as shown in Figure~\ref{fig:pump}. However, this is not always possible, as shown in Figure~\ref{fig:notpump}, where an attempt to pump would result in a glue mismatch, which would block the growth.

\begin{figure}[ht]
  \hspace*{\fill}
  \begin{tikzpicture}[scale=0.5]
    \begin{scope}[shift={(0.5,0.5)}]
      \draw[very thick,xgreen](0,0)--(1,0)--(1,1)--(4,1)--(4,0)--(4,-1)--(5,-1);
    \end{scope}
    \draws{0}{0}
    \draw(0.5,0.5)node{$\sigma$};
    \drawg{1}{0}
    \drawg{1}{1}
    \drawg{2}{1}
    \drawg{3}{1}
    \drawg{4}{1}
    \drawg{4}{0}
    \drawg{4}{-1}
    \drawg{5}{-1}
    \draw[ultra thick,red](2,1)rectangle(3,2);
    \draw[ultra thick,red](5,-1)rectangle(6,0);

    \begin{scope}[shift={(10,0)}]
      \begin{scope}[shift={(0.5,0.5)}]
        \draw[very thick,xgreen](0,0)--(1,0)--(1,1)--(4,1)--(4,0)--(4,-1)--(5,-1);
      \end{scope}
      \draws{0}{0}
      \draw(0.5,0.5)node{$\sigma$};
      \drawg{1}{0}
      \drawg{1}{1}
      \drawg{2}{1}
      \drawg{3}{1}
      \drawg{4}{1}
      \drawg{4}{0}
      \drawg{4}{-1}
      \drawg{5}{-1}

      \begin{scope}[shift={(3,-2)}]
        \begin{scope}[shift={(0.5,0.5)}]
          \draw[very thick, xgreen](2,1)--(4,1)--(4,-1)--(7,-1)--(7,-3)--(8,-3);
          \draw[very thick, xgreen, dashed](8,-3)--(9.2,-3);
        \end{scope}
        \drawg{3}{1}
        \drawg{4}{1}
        \drawg{4}{0}
        \drawg{4}{-1}
        \drawg{5}{-1}
        \draw[ultra thick,red](5,-1)rectangle(6,0);
        \begin{scope}[shift={(3,-2)}]
          \drawg{3}{1}
          \drawg{4}{1}
          \drawg{4}{0}
          \drawg{4}{-1}
          \drawg{5}{-1}
          \draw[ultra thick,red](2,1)rectangle(3,2);
        \end{scope}
      \end{scope}
      
      \draw[ultra thick,red](2,1)rectangle(3,2);
      \draw[ultra thick,red](5,-1)rectangle(6,0);
    \end{scope}
  \end{tikzpicture}
  \hspace*{\fill}
  \caption{Example of ``pumping'' a path. The seed is the white tile, and the two tiles highlighted in red are of the same type. We can try to repeat the part between them infinitely many times.}
  \label{fig:pump}
\end{figure}

\begin{figure}[ht]
  \hspace*{\fill}
  \begin{tikzpicture}[scale=0.5]
    \begin{scope}[shift={(0.5,0.5)}]
      \draw[very thick,xgreen](-1,0)--(1,0)--(2,0)--(2,1)--(1,1)--(1,2)--(4,2)--(4,0)--(5,0)--(5,0.5);
    \end{scope}
    \draws{-1}{0}
    \draw(-0.5,0.5)node{$\sigma$};
    \drawg{0}{0}
    \drawg{1}{0}
    \drawg{2}{0}
    \drawg{2}{1}
    \drawg{1}{1}
    \drawg{1}{2}
    \drawg{2}{2}
    \drawg{3}{2}
    \drawg{4}{2}
    \drawg{4}{1}
    \drawg{4}{0}
    \drawg{5}{0}
      
    \draw[ultra thick,red](5,0)rectangle(6,1);
    \draw[ultra thick,red](2,0)rectangle(3,1);
  \end{tikzpicture}
  \hspace*{\fill}
  \caption{An example path which, unlike the one in Figure~\ref{fig:pump}, cannot be repeated even a single time completely. The seed is the white tile, and the two tiles highlighted in red are of the same type.}
  \label{fig:notpump}
\end{figure}

Before this paper, a single \emph{positive} construction was known, in which for all $\varepsilon$, a tileset $T_\varepsilon$ could build multiple assemblies, all of Manhattan diameter $(2-\varepsilon)|T_\varepsilon|$ (this means in particular that $T_\varepsilon$ cannot build any infinite assembly).
Even though that result was the first example of an \emph{algorithmic} construction, the term ``algorithm'' in that case is to be taken in an extremely weak sense of a program whose running time is larger than its size.
Indeed, the resulting assemblies were only a constant factor bigger than the program size, like in a program where we call the same function twice.

Here, we show a way to build an assembly of width $\Omega(n\log n)$ with only $n$ different tile types, using the two dimensions to build a \emph{``controlled loop''}:

\newcommand\positivethm{
  For all $t\geq 0$, there is a tile assembly system $\mathcal T_t=(T_t,\sigma_t,1)$, where $|\sigma_t|=1$, $|T_t|=t$, and all assemblies $\alpha_t \in \termasm{\mathcal T_t}$ are of width $w_t \in ({t\log_3 t})/5 + O(t)$ and height less than $t$, and contain the same path $P_t$ 
  of width $w_t$.
}
\begin{theorem}
  \label{thm:positive}
  \positivethm
\end{theorem}

However, there are strong reasons to believe that 2D noncooperative tile assembly is not capable of performing Turing computation, since it is in particular not capable of simulating Turing machines inside a rectangle~\cite{STOC2017}, which is the only known form of Turing computation in tile assembly.

This result is therefore not meant as a first step towards ``full-featured algorithm'', but will be useful as a benchmark against which strategies to characterise the computational power of this model can be evaluated.

\section{Definitions and preliminaries}\label{sec:defs}
\label{definitions}

These definitions are for a large part taken from~\cite{STOC2017}. 



\subsection{Abstract tile assembly model}\label{sec:atam}

The abstract tile assembly model was introduced by Winfree~\cite{Winf98}. In this paper we study a restriction of this model called the temperature 1 abstract tile assembly model, or noncooperative abstract tile assembly model. For a mode detailed definition of the full model, as well as intuitive explanations, see for example~\cite{RotWin00,Roth01}.

A \emph{tile type} is a unit square with four sides, each consisting of a glue \emph{type} and a nonnegative integer \emph{strength}. Let  $T$  be a a finite set of tile types.
The sides of a tile type are respectively called  north, east, south, and west, as shown  in the following picture:
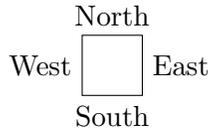
\begin{center}
\vspace{-1ex}
\begin{tikzpicture}[scale=0.8]
\draw(0,0)rectangle(1,1);
\draw(0,0.5)node[anchor=east]{West};
\draw(1,0.5)node[anchor=west]{East};
\draw(0.5,0)node[anchor=north]{South};
\draw(0.5,1)node[anchor=south]{North};
\end{tikzpicture}
\vspace{-1ex}\end{center}

An \emph{assembly} is a partial function $\alpha:\mathbb{Z}^2\dashrightarrow T$ where $T$ is a set of tile types and the domain of $\alpha$ (denoted $\dom{\alpha}$) is connected.
We let $\mathcal{A}^T$ denote the set of all assemblies over the set of tile types $T$. 
In this paper, two tile types in an assembly are said to  {\em bind} (or \emph{interact}, or are
\emph{stably attached}), if the glue types on their abutting sides are
equal, and have strength $\geq 1$.  An assembly $\alpha$ induces a
weighted \emph{binding graph} $G_\alpha=(V,E)$, where $V=\dom{\alpha}$, and
there is an edge $\{ a,b \} \in E$ if and only if the tiles at positions $a$ and $b$ interact, and
this edge is weighted by the glue strength of that interaction.  The
assembly is said to be $\tau$-stable if every cut of $G_\alpha$ has weight at
least $\tau$. A \emph{tile assembly system} is a triple $\mathcal{T}=(T,\sigma,\tau)$,
where $T$ is a finite set of tile types, $\sigma$ is a $\tau$-stable assembly called the \emph{seed}, and
$\tau \in \mathbb{N}$ is the \emph{temperature}.

Given two $\tau$-stable assemblies $\alpha$ and $\beta$, we say that $\alpha$ is a
\emph{subassembly} of $\beta$, and write $\alpha\sqsubseteq\beta$, if
$\dom{\alpha}\subseteq \dom{\beta}$ and for all $p\in \dom{\alpha}$,
$\alpha(p)=\beta(p)$.
We also write
$\alpha\rightarrow_1^{\mathcal{T}}\beta$ if we can obtain $\beta$ from
$\alpha$ by the binding of a single tile type, that is:  $\alpha\sqsubseteq \beta$, $|\dom{\beta}\setminus\dom{\alpha}|=1$ and the tile type at the position $\dom{\beta}\setminus\dom{\alpha}$ stably binds to $\alpha$ at that position.  We say that $\gamma$ is
\emph{producible} from $\alpha$, and write
$\alpha\rightarrow^{\mathcal{T}}\gamma$ if there is a (possibly empty)
sequence $\alpha_1,\alpha_2,\ldots,\alpha_n$ where $n \in \N \cup \{ \infty \} $, $\alpha= \alpha_1$ and $\alpha_n =\gamma$, such that
$\alpha_1\rightarrow_1^{\mathcal{T}}\alpha_2\rightarrow_1^{\mathcal{T}}\ldots\rightarrow_1^{\mathcal{T}}\alpha_n$. 
A sequence of $n\in\mathbb{Z}^+ \cup \{\infty\}$ assemblies
$\alpha_0,\alpha_1,\ldots$ over $\mathcal{A}^T$ is a
\emph{$\mathcal{T}$-assembly sequence} if, for all $1 \leq i < n$,
$\alpha_{i-1} \to_1^\mathcal{T} \alpha_{i}$.

The set of \emph{productions}, or \emph{producible assemblies}, of a tile assembly system $\mathcal{T}=(T,\sigma,\tau)$ is the set of all assemblies producible
from the seed assembly $\sigma$ and is written~$\prodasm{\mathcal{T}}$. An assembly $\alpha$ is called \emph{terminal} if there is no $\beta$ such that $\alpha\rightarrow_1^{\mathcal{T}}\beta$. The set of all terminal assemblies of $\mathcal{T}$ is denoted~$\termasm{\mathcal{T}}$.  

In this paper, we consider that $\tau=1$.  Thus, we make the simplifying assumption that all glue types have strength 0 or 1: it is not difficult to see that this assumption does not change the behavior of the model (if a glue type $g$ has strength $s_g \geq 1$, in the $\tau =1$ model then a tile with glue type $g$ binds to a matching glue type on an assembly border irrespective of the exact value of $s_g$). Consider a assembly $\alpha$ which is producible by a tile assembly system at temperature $1$, since only one glue of strenght $1$ is needed to stably bind a tile type to an assembly then any path of the binding graph of $\alpha$ can grow if it is bind to the seed. Thus at temperature $1$, it is more pertinent to consider path instead of assembly. Now, we introduce definitions which are useful to study temperature $1$.

\subsection{Paths and non-cooperative self-assembly}\label{sec:defs-paths}

Let $T$ be a set of tile types. A {\em tile} is a  pair $((x,y),t)  \in \mathbb{Z}^2 \times T$ where $(x,y)$ is a {\em position} and $t$ is a tile type. Intuitively, a path is a finite or one-way-infinite simple (non-self-intersecting) sequence of tiles placed on points of $\mathbb{Z}^2$ so that each tile in the sequence interacts with the previous one, or more precisely:

\begin{definition}[Path]\label{def:path}
  A {\em path} is a (finite or infinite) sequence  $P = P_0 P_1 P_2 \ldots$  of tiles   $P_i = ((x_i,y_i),t_i) \in \mathbb{Z}^2 \times T$, such that:
\begin{itemize}
\item for all $P_j$ and $P_{j+1}$ defined on $P$ it is the case that~$t_{j}$ and~$t_{j+1}$ interact, and
\item for all $P_j,P_k$ such that $j\neq k$ it is the case that $ (x_j,y_j) \neq (x_k,y_k)$.
\end{itemize}
\end{definition}

Whenever $P$ is finite, i.e. $P = P_0P_1P_2\ldots P_{n-1}$ for some $n$, $n$ is termed the {\em length} of $P$.
Note that by definition, paths are simple (or self-avoiding). 


Although a path is not an assembly, we know that each adjacent pair of tiles in the path sequence interact implying that the set of path positions forms a connected set in $\Z^2$ and hence every path uniquely represents an assembly containing exactly the tiles of the path. More formally, since an assembly is a function of $TODO$ then it can also interpreted as a subset of tiles and then for a path $P = P_0 P_1 P_2 \ldots$ we define the assembly  $\pathassembly{P} = \{ P_0, P_1, P_2, \ldots\}$ which we observe is  an assembly\footnote{or vice-versa, $\pathassembly{P}$ can be interpreted as a partial function from $\Z^2$ to tile types that is defined on a connected set.} and  we call $\pathassembly{P}$ a {\em path assembly}. 
A {\em path $P$ is said to be producible} by some tile assembly system $\calT = (T,\sigma,1)$ if the assembly $(\asm{P} \cup \sigma ) \in \prodT$ is producible, and we call such a $P$ a {\em producible path}. 
We define $$\prodpathsT = \{ P  \mid P \textrm{ is a path and } (\asm{P}\cup\sigma) \in \prodT \} $$  to be the set of producible paths of $\calT$.\footnote{Intuitively, although  producible paths  are not assemblies, any  producible path $P$ has the nice property that it encodes an unambiguous description of how to grow $\asm{P}$ from the seed $\sigma$, in ($P$) path  order, to produce  the assembly $\sigma \cup \asm{P}$.} 

For any path $P = P_0 P_1 P_2, \ldots$ and integer $i\geq 0$, we write $\pos{P_i} \in \mathbb{Z}^2$, or  $(x_{P_i},y_{P_i}) \in \mathbb{Z}^2$, for the position of $P_i$ and $\type{P_i}$ for the tile type of $P_i$. Hence if  $P_i = ((x_i,y_i),t_i) $ then $\pos{P_i} =  (x_{P_i},y_{P_i}) = (x_i,y_i) $ and $\type{P_i} = t_i$.

Note that, since the domain of a producible assembly is a connected set in $\Z^2$, and since in an assembly sequence of some tile assembly system $\calT = (T,\sigma,1)$ each tile binding event $\beta_i \rightarrow_1^\mathcal{T} \beta_{i+1} $ adds a single node $v$ to the binding graph $G_{\beta_{i}}$ of $\beta_i$  to give a new binding graph $G_{\beta_{i+1}}$, and adds at least one  weight-1 edge joining $v$ to the subgraph $G_{\beta_i} \in G_{\beta_{i+1}}$, then for any tile $((x,y),t) \in \alpha$  in a producible assembly $\alpha \in \prodT$, there is a edge-path (sequence of edges) in the binding graph of $\alpha$ from $\sigma$ to $((x,y),t)$. From there, the following important fact about temperature 1 tile  assembly is straightforward to see.

\begin{observation}
Let $\calT = (T,\sigma,1)$ be a tile assembly system and let $\alpha \in \prodT$. 
For any tile $((x,y),t) \in \alpha$ there is a producible path $P \in \prodpathsT$ that for some $i \in \N$ contains $P_i = ((x,y),t)$.
\end{observation}


When referring to the relative placements of positions in the grid graph of $\mathbb{Z}^2$, we say that a position $(x, y)$ is \emph{east of} (\resp, \emph{west of}, \emph{north}, \emph{south}) of another position $(x', y')$ if $x \geq x'$ (\resp $x \leq x'$, $y\geq y'$, $y\leq y'$).

If two paths, or two assemblies, or a path and an assembly, share a common position we say they {\em intersect} at that position. Furthermore, we say that two paths, or two assemblies, or a path and an assembly,  {\em agree} on a position if they both place the same tile type at that position and {\em conflict} if they place a different tile type at that position. We sometimes say that a path $P$ is {\em blockable} to mean that there is another path $P'$ (producible  by the same tile assembly system as produced $P$) that conflicts with $P$.

The translation of a tile $((x,y),t)$ by a vector $\vect{v}=(x',y')$ of $\mathbb{Z}^2$ is $((x+x',y+y'),t)$ (the type of the tile is not modified while its position is translated by $\vect{u})$. The translation of a path $P$ by $\vect{v}$, written $P+\vect{v}$, is the path $Q$ where
and for all indices $i$ of $P$,  $Q_i=P_i+\vect{v }$. As a convenient notation, for a path $PQ$ composed of subpaths $P$ and $Q$,  when we write $PQ +\vect{v} $ we mean $(PQ)+\vect{v}$ (i.e.\ the translation of all of $PQ$ by $+\vect{v}$).
The translation of a path $P$ by a vector $\vect{v} \in \mathbb{Z}^2$, written $P+\vect{v}$, is  the path $Q$ where
and for all indices $i$ of $P$,  


The \emph{width} of an assembly $\alpha$ is the number of columns on which $\alpha$ has at least one tile, and the \emph{height} of $\alpha$ is the number of rows on which $\alpha$ has at least one tile.

%



\section{The tile assembly system.}

\subsection{Definition of the tile assembly system.}
\label{long:sec:def}

Our construction relies on two parameters $\para, \parb \in \mathbb{N}$. Also, we define the series $(\series_i)_{i\geq 0}$ as: $$\series_0 = 2, \series_1=4 \text{ and for } i\geq 2, \series_i = 3\series_{i-1}-\series_{i-2}.$$ 
In this paper, we work on a zone of the $2D$ plane delimited as follow: we will only consider positions $(x,y) \in \mathbb{N}$ such that $0\leq x \leq (\step+1)\para-1$ and $0 \leq y \leq \series_\para-1$. The seed will be made of only one tile at position $(0,0)$ and any assembly will have a height bounded by $\series_\para-1$ and a width bounded by $(\para+1)\parb-1$.

The aim of this section is to define a the tile assembly system $\tasr{\para}=(\glueset,\sigma,1)$. The path of Figure \ref{fig:def} illustrates the definition of the set of tile types $\glueset$: each tile type is used exactly one time in this assembly. This path is made of five parts: one is the seed $\sigma$ and the four others are represented in green, orange, blue and red. The set of tile types $\glueset$ is defined as the union of these five kinds of tile types $\gluesets,\gluesetg,\gluesety,\gluesetb$ and $\gluesetp$. 

\vspace{+0.5em}
\noindent \emph{The seed.}
\vspace{+0.5em}

Consider the tile type $\tiles$ with only one glue called $\glueg_0$ on its east side. We define $\gluesets=\{\tiles\}$ and the seed $\sigma$ is defined as the assembly made of only the tile $((0,0), \tiles)$. From now on, $\sigma$ will always be the seed of our tile assembly system. 

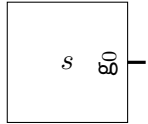
\newcommand\gluesize{0.3}
\begin{center}
  \begin{tikzpicture}[scale=0.8]
    \draw(0,0)rectangle(2,2);
    \draw(1,1)node{$\tiles$};
    \draw[very thick](2,1)--(2+\gluesize,1);
    \draw(2, 1)node[anchor=south,rotate=90]{$\glueg_0$};
  \end{tikzpicture}
\end{center}

\vspace{+0.5em}
\noindent \emph{The green tile types.}
\vspace{+0.5em}

The second kind of tile types $\gluesetg$ is made of $\parb+2$ tile types called $\tileg_0,\tileg_1, \ldots, \tileg_{\parb+1}$ defined as follow:

\begin{itemize}
\item for all $0\leq i \leq \parb-2$ the tile type $\tileg_i$ is made of the glue $\glueg_i$ on its west side and the glue $\glueg_{i+1}$ on its east side; moreover the tile $\tg_i$ is defined as $((1+i,0),\tileg_i)$.
\item the tile type $\tileg_{\parb-1}$ is made of the glue $\glueg_{\parb-1}$ on its west side and the glue $\glueg_{\parb}$ on its north side; moreover the tile $\tg_{\parb-1}$ is defined as $((\parb,0),\tileg_{\parb-1})$.
\item the tile type $\tileg_{\parb}$ is made of the glue $\glueg_{\parb}$ on its south side and the glue $\glueg_{\parb+1}$ on its west side; moreover the tile $\tg_{\parb}$ is defined as $((\parb,1),\tileg_{\parb})$.
\item the tile type $\tileg_{\parb+1}$ is made of the glue $\glueg_{\parb+1}$ on its east side and the glue $\gluey_{0}$ on its north side; moreover the tile $\tg_{\parb+1}$ is defined as $((\parb-1,1),\tileg_{\parb+1})$.
\end{itemize}
This set of tile types is used to hardcode the path $\pathg=\tg_0\tg_1 \ldots \tg_{\parb+1}$. 

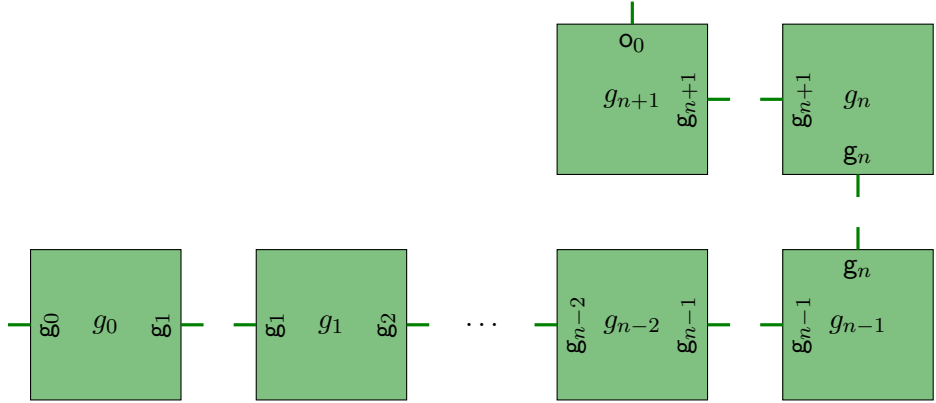
\begin{center}
  \begin{tikzpicture}[scale=1]
    \draw[fill=xgreen!50!white](0,0)rectangle(2,2);
    \draw[very thick,xgreen](0,1)--(-\gluesize,1);
    \draw[very thick,xgreen](2,1)--(2+\gluesize,1);
    \draw(1,1)node{$\tileg_0$};
    \draw(0,1)node[anchor=north,rotate=90]{$\glueg_0$};
    \draw(2,1)node[anchor=south,rotate=90]{$\glueg_1$};
    \begin{scope}[shift={(3,0)}]
      \draw[fill=xgreen!50!white](0,0)rectangle(2,2);
      \draw[very thick,xgreen](0,1)--(-\gluesize,1);
      \draw[very thick,xgreen](2,1)--(2+\gluesize,1);
      \draw(1,1)node{$\tileg_1$};
      \draw(0,1)node[anchor=north,rotate=90]{$\glueg_1$};
      \draw(2,1)node[anchor=south,rotate=90]{$\glueg_2$};
      \draw(3,1)node{$\ldots$};
      \begin{scope}[shift={(4,0)}]
        \draw[fill=xgreen!50!white](0,0)rectangle(2,2);
        \draw(1,1)node{$\tileg_{n-2}$};
        \draw(0,1)node[anchor=north,rotate=90]{$\glueg_{n-2}$};
        \draw(2,1)node[anchor=south,rotate=90]{$\glueg_{n-1}$};
        \draw[very thick,xgreen](0,1)--(-\gluesize,1);
        \draw[very thick,xgreen](2,1)--(2+\gluesize,1);
        \begin{scope}[shift={(3,0)}]
          \draw[fill=xgreen!50!white](0,0)rectangle(2,2);
          \draw(1,1)node{$\tileg_{n-1}$};
          \draw(0,1)node[anchor=north,rotate=90]{$\glueg_{n-1}$};
          \draw(1,2)node[anchor=north]{$\glueg_{n}$};
          \draw[very thick,xgreen](1,2)--(1, 2+\gluesize);
          \draw[very thick,xgreen](0,1)--(-\gluesize,1);
        \end{scope}
        \begin{scope}[shift={(3,3)}]
          \draw[fill=xgreen!50!white](0,0)rectangle(2,2);
          \draw(1,1)node{$\tileg_{n}$};
          \draw(1,0)node[anchor=south]{$\glueg_n$};
          \draw(0,1)node[anchor=north,rotate=90]{$\glueg_{n+1}$};
          \draw[very thick,xgreen](1,0)--(1, -\gluesize);
          \draw[very thick,xgreen](0,1)--(-\gluesize,1);
        \end{scope}
        \begin{scope}[shift={(0,3)}]
          \draw[fill=xgreen!50!white](0,0)rectangle(2,2);
          \draw(1,1)node{$\tileg_{n+1}$};
          \draw(2,1)node[anchor=south,rotate=90]{$\glueg_{n+1}$};
          \draw(1,2)node[anchor=north]{$\gluey_0$};
          \draw[very thick,xgreen](1,2)--(1, 2+\gluesize);
          \draw[very thick,xgreen](2,1)--(2+\gluesize,1);
        \end{scope}
      \end{scope}
    \end{scope}
  \end{tikzpicture}
 \end{center}

\vspace{+0.5em}
\noindent \emph{The orange tile types.}
\vspace{+0.5em}

The third kind of tile types $\gluesety$ is made of $\series_\para-2$ tile types called $\tiley_0,\tiley_1, \ldots, \tiley_{\series_\para-3}$ defined as follow. For all $0\leq i \leq \series_\para-3$ the tile type $\tiley_i$ is made of:
\begin{itemize}
\item the glue $\gluey_i$ on its south side;
\item the glue $\gluey_{i+1}$ on its north side if and only if $i < \series_\para-3$;
\item the glue $\glueb_0$ on its east side if and only if there exists $1\leq \step \leq \para$ such that $i = \series_\step-3$.
\end{itemize}
For all $0\leq i \leq \series_\para-3$, the tile $\ty_i$ is defined as $((\parb-1,2+i),\tiley_i)$. This set of tile types is used to hardcode the paths $\pathy{\step} = \ty_0 \ty_1 \ldots \ty_{\series_\step-3}$ for $1 \leq \step \leq \para$. Note that, for all $1\leq \step \leq \para-1$, the path $\pathy{\step}$ is a prefix of the path $\pathy{{\step+1}}$.  To summarize, for all $i\in\{0,1,\ldots,\series_\para-3\}$:

\begin{itemize}
  \item if $i=\series_{\para-3}$ then $\tiley_{i}$ is the following tile type:
    
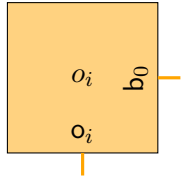
\begin{center}
      \begin{tikzpicture}[scale=1]
        \draw[fill=xorange!50!white](0,0)rectangle(2,2);
        \draw(1,1)node{$\tiley_{i}$};
        \draw(1,0)node[anchor=south]{$\gluey_{i}$};
        \draw(2,1)node[anchor=south, rotate=90]{$\glueb_0$};
        \draw[very thick,xorange](1,0)--(1, -\gluesize);
        \draw[very thick,xorange](2,1)--(2+\gluesize,1);
      \end{tikzpicture}
    \end{center}

  \item if $i = \series_j-3$ for some integer $1\leq j < \para$, then $\tiley_i$ is the following tile type:
    
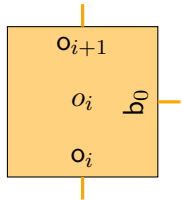
\begin{center}
      \begin{tikzpicture}[scale=1]
        \draw[fill=xorange!50!white](0,0)rectangle(2,2);
        \draw(1,1)node{$\tiley_{i}$};
        \draw(1,0)node[anchor=south]{$\gluey_{i}$};
        \draw(1,2)node[anchor=north]{$\gluey_{i+1}$};
        \draw(2,1)node[anchor=south, rotate=90]{$\glueb_0$};
        \draw[very thick,xorange](1,2)--(1, 2+\gluesize);
        \draw[very thick,xorange](1,0)--(1, -\gluesize);
        \draw[very thick,xorange](2,1)--(2+\gluesize,1);
      \end{tikzpicture}
    \end{center}

  \item otherwise, $\tiley_i$ is the following tile type:
    
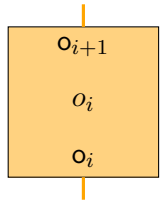
\begin{center}
      \begin{tikzpicture}[scale=1]
        \draw[fill=xorange!50!white](0,0)rectangle(2,2);
        \draw(1,1)node{$\tiley_{i}$};
        \draw(1,0)node[anchor=south]{$\gluey_i$};
        \draw(1,2)node[anchor=north]{$\gluey_{i+1}$};
        \draw[very thick,xorange](1,2)--(1, 2+\gluesize);
        \draw[very thick,xorange](1,0)--(1, -\gluesize);
      \end{tikzpicture}
    \end{center}
\end{itemize}

\vspace{+0.5em}
\noindent \emph{The blue tile types.}
\vspace{+0.5em}

The fourth kind of tile types $\gluesetb$ is made of $2\parb$ tile types called $\tileb_0,\tileb_1, \ldots, \tileb_{2\parb-1}$ defined as follow:

\begin{itemize}
\item for all $0\leq i \leq \parb-2$ the tile type $\tileb_i$ is made of the glue $\glueb_i$ on its west side and the glue $\glueb_{i+1}$ on its east side; moreover the tile $\tb_i$ is defined as $((\parb+i,3),\tileb_i)$;
\item the tile type $\tileb_{\parb-1}$ is made of the glue $\glueb_{\parb-1}$ on its west side and the glue $\glueb_{\parb}$ on its south side; moreover the tile $\tb_{\parb-1}$ is defined as $((2\parb-1,3),\tileb_{\parb-1})$;
\item the tile type $\tileb_{\parb}$ is made of the glue $\glueb_{\parb}$ on its north side and the glue $\glueb_{\parb+1}$ on its west side; moreover the tile $\tb_{\parb}$ is defined as $((2\parb-1,2),\tileb_{\parb})$;
\item for all $\parb+1\leq i \leq 2\parb-2$, the tile type $\tileb_{i}$ is made of the glue $\glueb_{i}$ on its east side and the glue $\glueb_{i+1}$ on its west side; moreover the tile $\tb_{i}$ is defined as $((3\parb-1-i,2),\tileb_{i})$;
\item the tile type $\tileb_{2\parb-1}$ is made of the glue $\glueb_{2\parb-1}$ on its east side and the glue $\gluep_{0}$ on its south side; moreover the tile $\tb_{2\parb-1}$ is defined as $((\parb,2),\tileb_{2\parb-1})$.
\end{itemize}
This set of tile types is used to hardcode the path $\pathb{1}=\tb_0 \tb_1 \ldots \tb_{2n-1}$. For all $2\leq \step \leq \para$, we define the path $\pathb{\step}$ as $\pathb{1}$ translated by $(0,\series_{\step}-4)$. 

  
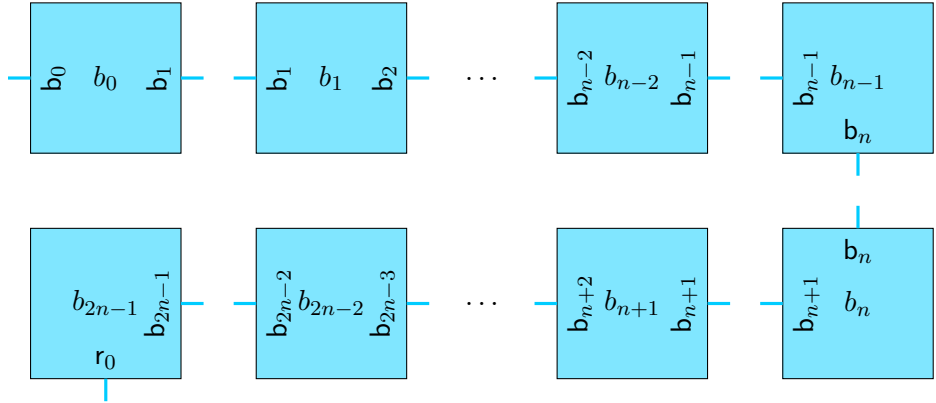
\begin{center}
    \begin{tikzpicture}[scale=1]
      \draw[fill=xblue!50!white](0,0)rectangle(2,2);
      \draw(1,1)node{$\tileb_{0}$};
      \draw(0,1)node[anchor=north,rotate=90]{$\glueb_0$};
      \draw(2,1)node[anchor=south,rotate=90]{$\glueb_1$};
      \draw[very thick,xblue](0,1)--(-\gluesize,1);
      \draw[very thick,xblue](2,1)--(2+\gluesize,1);
      \begin{scope}[shift={(3,0)}]
        \draw[fill=xblue!50!white](0,0)rectangle(2,2);
        \draw(1,1)node{$\tileb_{1}$};
        \draw(0,1)node[anchor=north,rotate=90]{$\glueb_1$};
        \draw(2,1)node[anchor=south,rotate=90]{$\glueb_2$};
        \draw[very thick,xblue](0,1)--(-\gluesize,1);
        \draw[very thick,xblue](2,1)--(2+\gluesize,1);
      \end{scope}
      \draw(6,1)node{$\ldots$};
      \begin{scope}[shift={(7,0)}]
        \draw[fill=xblue!50!white](0,0)rectangle(2,2);
        \draw(1,1)node{$\tileb_{n-2}$};
        \draw(0,1)node[anchor=north,rotate=90]{$\glueb_{\parb-2}$};
        \draw(2,1)node[anchor=south,rotate=90]{$\glueb_{\parb-1}$};
        \draw[very thick,xblue](0,1)--(-\gluesize,1);
        \draw[very thick,xblue](2,1)--(2+\gluesize,1);
      \end{scope}
      \begin{scope}[shift={(10,0)}]
        \draw[fill=xblue!50!white](0,0)rectangle(2,2);
        \draw(1,1)node{$\tileb_{n-1}$};
        \draw(0,1)node[anchor=north,rotate=90]{$\glueb_{\parb-1}$};
        \draw(1,0)node[anchor=south]{$\glueb_{n}$};
        \draw[very thick,xblue](0,1)--(-\gluesize,1);
        \draw[very thick,xblue](1,0)--(1, -\gluesize);
      \end{scope}
      \begin{scope}[shift={(10,-3)}]
        \draw[fill=xblue!50!white](0,0)rectangle(2,2);
        \draw(1,1)node{$\tileb_{\parb}$};
        \draw(1,2)node[anchor=north]{$\glueb_n$};
        \draw(0,1)node[anchor=north,rotate=90]{$\glueb_{\parb+1}$};
        \draw[very thick,xblue](0,1)--(-\gluesize,1);
        \draw[very thick,xblue](1,2)--(1,2+\gluesize);
      \end{scope}

      \begin{scope}[shift={(0,-3)}]
        \draw[fill=xblue!50!white](0,0)rectangle(2,2);
        \draw(1,0)node[anchor=south]{$\gluep_0$};
        \draw(1,1)node{$\tileb_{2\parb-1}$};
        \draw(2,1)node[anchor=south,rotate=90]{$\glueb_{2\parb-1}$};
        \draw[very thick,xblue](1,0)--(1,-\gluesize);
        \draw[very thick,xblue](2,1)--(2+\gluesize,1);
        \begin{scope}[shift={(3,0)}]
          \draw[fill=xblue!50!white](0,0)rectangle(2,2);
          \draw(1,1)node{$\tileb_{2\parb-2}$};
          \draw(0,1)node[anchor=north,rotate=90]{$\glueb_{2\parb-2}$};
          \draw(2,1)node[anchor=south,rotate=90]{$\glueb_{2\parb-3}$};
          \draw[very thick,xblue](0,1)--(-\gluesize,1);
          \draw[very thick,xblue](2,1)--(2+\gluesize,1);
        \end{scope}
        \draw(6,1)node{$\ldots$};
        \begin{scope}[shift={(7,0)}]
          \draw[fill=xblue!50!white](0,0)rectangle(2,2);
          \draw(1,1)node{$\tileb_{\parb+1}$};
          \draw(0,1)node[anchor=north,rotate=90]{$\glueb_{\parb+2}$};
          \draw(2,1)node[anchor=south,rotate=90]{$\glueb_{\parb+1}$};
          \draw[very thick,xblue](0,1)--(-\gluesize,1);
          \draw[very thick,xblue](2,1)--(2+\gluesize,1);
        \end{scope}
      \end{scope}
    \end{tikzpicture}
  \end{center}

\vspace{+0.5em}
\noindent \emph{The red tile types.}
\vspace{+0.5em}

The fifth kind of tile types $\gluesetp$ is made of $\series_\para-\series_{\para-1}-2$ tile types called $\tilep_0,\tilep_1, \ldots, \tilep_{\series_\para-\series_{\para-1}-3}$ defined as follow. For all $0\leq i \leq \series_\para-\series_{\para-1}-3$, the tile type $\tilep_i$ is made of:
\begin{itemize}
\item the glue $\gluep_i$ on its north side;
\item the glue $\gluep_{i+1}$ on its south side if and only if $i < \series_\para-\series_{\para-1}-3$;
\item the glue $\glueg_0$ on its east side if and only if there exists $2\leq \step \leq \para$ such that $i=\series_{\step}-\series_{\step-1}-3$.
\end{itemize}
For all $0\leq i \leq \series_\para-\series_{\para-1}-3$, the tile $\tp_i$ is defined as $((\parb,1-i),\tilep_i)$. This set of tile types is used to hardcode the path $\pathpurple$ defined as $\tp_0 \tp_1 \ldots \tp_{\series_\para-\series_{\para-1}-3}$. For all $2\leq i \leq \para$ and  $1\leq j \leq \para$, we define the path $\pathpp{j}{i}$ as the prefix of $\pathpurple$ of length $\series_{j}-\series_{j-1}-2$ translated by $(0,\series_{i}-4)$. Note that, for all $1\leq j' \leq j$, the path $\pathpp{j'}{i}$ is a prefix of $\pathpp{j}{i}$ and that for all $1\leq i \leq \para$, $\pathpp{1}{i}$ is $\epsilon$: the empty path of length $0$. To summarize, for all $i\in\{0,1,\ldots,\series_\para-\series_{\para-1}-3\}$:

\begin{itemize}
\item if $i={\series_\para-\series_{\para-1}-3}$ then $\tilep_i$ is the following tile type:

  
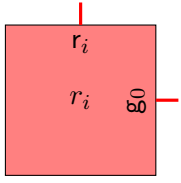
\begin{center}
    \begin{tikzpicture}[scale=1]
      \draw[fill=xred!50!white](0,0)rectangle(2,2);
      \draw(1,1)node{$\tilep_{i}$};
      \draw(1,2)node[anchor=north]{$\gluep_{i}$};
      \draw(2,1)node[anchor=south, rotate=90]{$\glueg_0$};
      \draw[very thick,xred](1,2)--(1, 2+\gluesize);
      \draw[very thick,xred](2,1)--(2+\gluesize,1);
    \end{tikzpicture}
  \end{center}

\item if $i = \series_j - \series_{j-1}-3$ for some integer $1\leq j< \para$, then $\tilep_i$ is the following tile type:

  
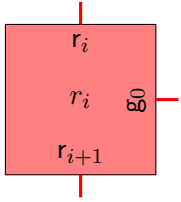
\begin{center}
    \begin{tikzpicture}[scale=1]
      \draw[fill=xred!50!white](0,0)rectangle(2,2);
      \draw(1,1)node{$\tilep_{i}$};
      \draw(1,0)node[anchor=south]{$\gluep_{i+1}$};
      \draw(1,2)node[anchor=north]{$\gluep_{i}$};
      \draw(2,1)node[anchor=south, rotate=90]{$\glueg_0$};
      \draw[very thick,xred](1,2)--(1, 2+\gluesize);
      \draw[very thick,xred](1,0)--(1, -\gluesize);
      \draw[very thick,xred](2,1)--(2+\gluesize,1);
    \end{tikzpicture}
  \end{center}
\item  otherwise, $\tilep_i$ is the following tile type:

  
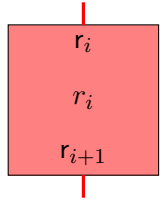
\begin{center}
    \begin{tikzpicture}[scale=1]
      \draw[fill=xred!50!white](0,0)rectangle(2,2);
      \draw(1,1)node{$\tilep_{i}$};
      \draw(1,0)node[anchor=south]{$\gluep_{i+1}$};
      \draw(1,2)node[anchor=north]{$\gluep_{i}$};
      \draw[very thick,xred](1,2)--(1, 2+\gluesize);
      \draw[very thick,xred](1,0)--(1, -\gluesize);
    \end{tikzpicture}
  \end{center}
\end{itemize}
\begin{figure}
\centering

\begin{tikzpicture}[x=0.25cm,y=0.25cm]

\path [draw, gray, very thin] (-4,-2) grid[step=0.25cm] (52,70);

 \draw[->] (-4,0.5) -- (52,0.5);
 \draw[->] (0.5,-2) -- (0.5,70);
 \draw[dotted,thick] (-4,68) -- (52,68);
 \draw[dotted,thick] (20,-2) -- (20,70);

\draw (-1.8,3.5) node {$\series_{1}-1$};
 \draw(0.33,3.5) -- (0.66,3.5);
\draw (-1.8,9.5) node {$\series_{2}-1$};
 \draw(0.33,9.5) -- (0.66,9.5);
\draw (-1.8,25.5) node {$\series_{3}-1$};
 \draw(0.33,25.5) -- (0.66,25.5);
\draw (-1.8,67.5) node {$\series_{4}-1$};
 \draw(0.33,67.5) -- (0.66,67.5);
\draw (10.5,-1) node {$\parb$};
\draw (19.5,-1) node {$2\parb-1$};
 \draw(19.5,0.33) -- (19.5,0.66);
\draw (29.5,-1) node {$3\parb-1$};
 \draw(29.5,0.33) -- (29.5,0.66);
\draw (39.5,-1) node {$4\parb-1$};
 \draw(39.5,0.33) -- (39.5,0.66);
\draw (49.5,-1) node {$5\parb-1$};
 \draw(49.5,0.33) -- (49.5,0.66);

 \draw[very thick,xgreen] (0.5,0.5) -| (10.5,1.5) -| (9.5,1.5);
 \draw[very thick,xorange] (9.5,1.5) |- (9.5,67.5);
 \draw[very thick,xblue] (9.5,3.5) |- (10,3.5);
 \draw[very thick,xblue] (9.5,9.5) |- (10,9.5);
 \draw[very thick,xblue] (9.5,25.5) -| (10,25.5);
 \draw[very thick,xblue] (9.5,67.5) -| (19.5,66.5) -| (10.5,66.5);
 \draw[very thick,xred] (10.5,66.5) -| (10.5,26.5);
 \draw[very thick,xgreen] (10.5,26.5) -| (11,26.5);
 \draw[very thick,xgreen] (10.5,52.5) -| (11,52.5);
 \draw[very thick,xgreen] (10.5,62.5) -| (11,62.5);

\draws{0}{0}
\drawg{1}{0}
\drawg{2}{0}
\drawg{3}{0}
\drawg{4}{0}
\drawg{5}{0}
\drawg{6}{0}
\drawg{7}{0}
\drawg{8}{0}
\drawg{9}{0}
\drawg{10}{0}
\drawg{10}{1}
\drawg{9}{1}

\drawo{9}{2}
\drawo{9}{3}
\drawo{9}{4}
\drawo{9}{5}
\drawo{9}{6}
\drawo{9}{7}
\drawo{9}{8}
\drawo{9}{9}
\drawo{9}{10}
\drawo{9}{11}
\drawo{9}{12}
\drawo{9}{13}
\drawo{9}{14}
\drawo{9}{15}
\drawo{9}{16}
\drawo{9}{17}
\drawo{9}{18}
\drawo{9}{19}
\drawo{9}{20}
\drawo{9}{21}
\drawo{9}{22}
\drawo{9}{23}
\drawo{9}{24}
\drawo{9}{25}
\drawo{9}{26}
\drawo{9}{27}
\drawo{9}{28}
\drawo{9}{29}
\drawo{9}{30}
\drawo{9}{31}
\drawo{9}{32}
\drawo{9}{33}
\drawo{9}{34}
\drawo{9}{35}
\drawo{9}{36}
\drawo{9}{37}
\drawo{9}{38}
\drawo{9}{39}
\drawo{9}{40}
\drawo{9}{41}
\drawo{9}{42}
\drawo{9}{43}
\drawo{9}{44}
\drawo{9}{45}
\drawo{9}{46}
\drawo{9}{47}
\drawo{9}{48}
\drawo{9}{49}
\drawo{9}{50}
\drawo{9}{51}
\drawo{9}{52}
\drawo{9}{53}
\drawo{9}{54}
\drawo{9}{55}
\drawo{9}{56}
\drawo{9}{57}
\drawo{9}{58}
\drawo{9}{59}
\drawo{9}{60}
\drawo{9}{61}
\drawo{9}{62}
\drawo{9}{63}
\drawo{9}{64}
\drawo{9}{65}
\drawo{9}{66}
\drawo{9}{67}

\drawb{10}{67}
\drawb{11}{67}
\drawb{12}{67}
\drawb{13}{67}
\drawb{14}{67}
\drawb{15}{67}
\drawb{16}{67}
\drawb{17}{67}
\drawb{18}{67}
\drawb{19}{67}
\drawb{19}{66}
\drawb{18}{66}
\drawb{17}{66}
\drawb{16}{66}
\drawb{15}{66}
\drawb{14}{66}
\drawb{13}{66}
\drawb{12}{66}
\drawb{11}{66}
\drawb{10}{66}

\drawr{10}{65}
\drawr{10}{64}
\drawr{10}{63}
\drawr{10}{62}
\drawr{10}{61}
\drawr{10}{60}
\drawr{10}{59}
\drawr{10}{58}
\drawr{10}{57}
\drawr{10}{56}
\drawr{10}{55}
\drawr{10}{54}
\drawr{10}{53}
\drawr{10}{52}
\drawr{10}{51}
\drawr{10}{50}
\drawr{10}{49}
\drawr{10}{48}
\drawr{10}{47}
\drawr{10}{46}
\drawr{10}{45}
\drawr{10}{44}
\drawr{10}{43}
\drawr{10}{42}
\drawr{10}{41}
\drawr{10}{40}
\drawr{10}{39}
\drawr{10}{38}
\drawr{10}{37}
\drawr{10}{36}
\drawr{10}{35}
\drawr{10}{34}
\drawr{10}{33}
\drawr{10}{32}
\drawr{10}{31}
\drawr{10}{30}
\drawr{10}{29}
\drawr{10}{28}
\drawr{10}{27}
\drawr{10}{26}

\draw[fill=black] (9.5,3.5) circle (0.12);
\draw[fill=black] (9.5,9.5) circle (0.12);
\draw[fill=black] (9.5,25.5) circle (0.12);
\draw[fill=black] (9.5,67.5) circle (0.12);
\draw[fill=black] (10.5,62.5) circle (0.12);
\draw[fill=black] (10.5,52.5) circle (0.12);
\draw[fill=black] (10.5,26.5) circle (0.12);
\end{tikzpicture}

\caption{In our examples, we consider $\para=4$ and $\parb=10$. The seed (in white) is at position $(0,0)$ and we represent the path $\paths{4}{4}=\pathg\pathy{4}\pathb{4}\pathpp{4}{4}$. This path is producible by $\tas$ and this figure contains exactly one occurrence of each tile type of $\glueset$. The height of $\paths{4}{4}$ is $\series_{\para}-1$ and its width is $2\parb-1$. The tiles of $\pathy{4}$ and $\pathpp{4}{4}$ with a glue on their east side are marked by a black dot.}
\label{fig:def}
\end{figure}

\subsection{Basic properties.}
\label{long:sec:basic}

The aim of this section is to define a set of paths which characterized all the possible prefixes of a path producible by $\tas$. These paths are obtained by gluing together the different paths defined in section \ref{long:sec:def}. When two of these paths are glued together, we have to verify that the result of this operation is also a path. To achieve this goal, we have to check two properties. The first one is that the last tile of the first path can be glued to the first tile of the second path. The second one is that that the two paths do not intersect. We start by a first lemma which gives the positions occupied by the paths defined in section \ref{long:sec:def}. This lemma is useful to show that a path is west or north of another one and thus that these two paths do not intersect. 

\begin{lemma}
\label{lem:position}
For any $\para,\parb \in \mathbb{N}$, for any $1 \leq j\leq i \leq \para$ and for any position $(x,y)$ occupied by:
\begin{itemize}
\item a tile of $\pathg$, we have $1\leq x \leq \parb$ and $0\leq y\leq 1$;
\item a tile of $\pathy{i}$, we have $x=\parb-1$ and $2 \leq y \leq \series_i-1$;
\item a tile of $\pathb{i}$, we have $\parb \leq x \leq 2\parb-1$ and $\series_{i}-2 \leq y \leq \series_{i}-1$;
\item a tile of $\pathpp{j}{i}$, we have $x=\parb$ and $\series_{i}-\series_{j}+\series_{j-1} \leq y \leq \series_{i}-3$ 
(for $j=\step$, we have $\series_{i-1} \leq y \leq \series_{i}-3$).
\end{itemize}
\end{lemma}

\begin{proof}
Straightforward for the green, orange and blue paths. For any position $(x,y)$ occupied by $\pathpp{j}{i}$, we have $x=\parb$. Moreover, the position of $\tp_0+(0,\series_{\step}-4)$ is $(\parb,\series_{\step}+3)$ and the following tiles are all below the previous one and the length of $\pathpp{j}{i}$ is $\series_{j}-\series_{j-1}-2$ then the tile $\tp_{\series_{j}-\series_{j-1}-3}$ occupies the position $(\parb,\series_{i}-\series_{j}+\series_{j-1})$. For the special case where $j=i$, we have $\series_{i}-\series_{i}+\series_{i-1}=\series_{i-1}$.
\end{proof}

Now, for all $1\leq j\leq i \leq \para$, we defined $\paths{i}{j}$ as $\pathg\pathy{i}\pathb{i}\pathpp{j}{i}$ and we show that these sequences of tiles are paths producible by $\tas$ (see Figure \ref{fig:def} for an example of such a path).

\begin{lemma}
\label{lem:producible}
For any $\para,\parb \in \mathbb{N}$ and for any $1 \leq i \leq \para$, $\paths{i}{i}$ is a path producible by $\tas$.
\end{lemma}

\begin{proof}
Consider $1\leq i \leq \patha$ and the paths $\pathg$, $\pathy{i}$, $\pathb{i}$ and $\pathp{i}$. Firstly, we show that these different paths do no intersect. For the special where $\step=1$, we remind that $\pathg\pathy{1}\pathb{1}\pathp{1}=\pathg\pathy{1}\pathb{1}$. By lemma \ref{lem:position}:
\begin{itemize}
\item the seed $\sigma$ is west of $\pathg\pathy{i}\pathb{i}\pathp{i}$;
\item the path $\pathg$ is south of $\pathy{i}\pathb{i}\pathp{i}$;
\item the path $\pathy{i}$ is west of $\pathb{i}\pathp{i}$;
\item the path $\pathb{i}$ is north of $\pathp{i}$.
\end{itemize}
Secondly, we show that these different paths can be glued together:
\begin{itemize}
\item the position of the seed is $(0,0)$ and there is a glue $\tileg_0$ on its east side and the position of $\pathg_0$ is $(1,0)$ and there is a glue $\tileg_0$ on its west side;
\item the position of $\pathg_{\parb+1}$ is $(\parb-1,1)$ and there is a glue $\tiley_0$ on its north side and the position of $\pathy{i}_0$ is $(\parb-1,2)$ and there is a glue $\tiley_0$ on its south side;
\item the position of $\pathy{i}_{\series_{i}-3}$ is $(\parb-1,\series_{i}-1)$ and there is a glue $\tileb_0$ on its east side and the position of $\pathb{i}_0$ is $(\parb,\series_{i}-1)$ and there is a glue $\tileb_0$ on its west side;
\item the position of $\pathb{i}_{2\parb-1}$ is $(\parb,\series_{i}-2)$ and there is a glue $\tilep_0$ on its south side and the position of $\pathp{i}_0$ is $(\parb,\series_{i}-3)$ and there is a glue $\tilep_0$ on its north side.
\end{itemize}
Then, the path $\pathg\pathy{i}\pathb{i}\pathp{i}$ is producible by $\tas$. For all $1\leq i \leq \step$, $\paths{i}{i}$ is a prefix of $\paths{i}{i}$ and thus it is a path producible by $\tas$.
\end{proof}

Note that for any $1\leq i \leq \para$, all prefixes of $\paths{i}{i}$ are also producible by $\tas$. The paths $(\paths{\step}{i})_{1\leq i \leq \step \leq \para}$ will be used to characterize all the path producible by $\tas$. To achieve this goal, we need to know the positions of the free glues on these paths (see Figure \ref{fig:def}).

\begin{lemma}
\label{lem:freeglue}
For any $\para,\parb \in \mathbb{N}$ and for any $1 \leq i \leq \para$, the free glues of $\paths{i}{i}$ are:
\begin{itemize}
\item the north glue $\gluey_{\series_i-2}$ of the tile $\ty_{\series_i-3}$ whose position is $(\parb-1,\series_i-1)$ if $i<\para$;
\item for all $1\leq \inds \leq \step-1$ the east glue $\glueb_0$ of the tile $\ty_{\series_\inds-3}$ whose position is $(\parb-1,\series_\inds-1)$;
\item for all $2\leq j \leq \step$ the east glue $\glueg_0$ of the tile $\pathp{i}_{\series_{j}-\series_{j-1}-3}$ whose position is $(\parb,\series_{i}-\series_{j}+\series_{j-1})$;
\item the south glue $\gluep_{\series_{i}-\series_{i-1}-2}$ of the tile $\pathp{i}_{\series_{i}-\series_{i-1}-3}$ whose position is $(\parb,\series_{i-1})$ if $1<i<\para$.
\end{itemize}
\end{lemma}

\begin{proof}
Except for the last tile, for any tile of $\paths{\step}{\step}$ two of its glues are used to assemble the path $\paths{\step}{\step}$. Since the tile types have two or three glues on their side, we have to check the remaining glue on the tiles with three glues on their sides. First all the tile types of $\gluesetg$ and $\gluesetb$ have only two glues on their sides, thus there are no free glues on the tiles of $\pathg$ and $\pathb{\step}$. Note that for the special case where $\step=1$ (see Figure \ref{fig:path11}) the last tile of $\pathg\pathy{1}\pathb{1}\pathp{1}$ is the last tile of $\pathb{1}$. Nevertheless, the south glue of the tile $\pathb{1}_{2\parb-1}$ whose position is $(\parb,2)$ is not free because of a mismatch with the tile $\pathg_{\parb}$ whose position is $(\parb,1)$ and which has no glue on its north side. The tile types of $\gluesety$ with three glues on their side are $\tiley_{\series_\inds-3}$ for all $1\leq \inds < \para$. For all $1\leq \inds<i$, the glue on the east side of $\pathy{i}_{\series_\inds-3}$ are free whereas if $i<\para$  the north glue of tile $\pathy{i}_{\series_i-3}$ is free. Finally, the only tile types of $\gluesetp$ with three glues on their sides are $\tilep_{\series_{j}-\series_{j-1}-3}$ for all $2 \leq j \leq \para$. For $2\leq j \leq \step$, the glue on the east side of $\pathp{\step}_{\series_{j}-\series_{j-1}-3}$ is free. Moreover since $\pathp{i}_{\series_{i}-\series_{i-1}-3}$ is the last tile of the path then its south glue is also free (except when $\step=\para$ because $\tilep_{\series_{\para}-\series_{\para-1}-3}$ has no glue on its south side).
\end{proof}

\begin{figure}
\centering

\begin{tikzpicture}[x=0.25cm,y=0.25cm]

\path [draw, gray, very thin] (-4,-2) grid[step=0.25cm] (22,6);

 \draw[->] (-4,0.5) -- (22,0.5);
 \draw[->] (0.5,-2) -- (0.5,6);
 \draw[dotted,thick] (-4,4) -- (22,4);
 \draw[dotted,thick] (20,-2) -- (20,6);

\draw (-1.8,3.5) node {$\series_{1}-1$};
 \draw(0.33,3.5) -- (0.66,3.5);
\draw (10.5,-1) node {$\parb$};
\draw (19.5,-1) node {$2\parb-1$};
 \draw(19.5,0.33) -- (19.5,0.66);

 \draw[very thick,xgreen] (0.5,0.5) -| (10.5,1.5) -| (9.5,1.5);
 \draw[very thick,xorange] (9.5,1.5) |- (9.5,4);
 \draw[very thick,xblue] (9.5,3.5) -| (19.5,3.5) |- (10.5,2.5);
 \draw[very thick,xred] (10.5,2.5) -| (10.5,2);

\draws{0}{0}
\drawg{1}{0}
\drawg{2}{0}
\drawg{3}{0}
\drawg{4}{0}
\drawg{5}{0}
\drawg{6}{0}
\drawg{7}{0}
\drawg{8}{0}
\drawg{9}{0}
\drawg{10}{0}
\drawg{10}{1}
\drawg{9}{1}

\drawo{9}{2}
\drawo{9}{3}

\drawb{10}{3}
\drawb{11}{3}
\drawb{12}{3}
\drawb{13}{3}
\drawb{14}{3}
\drawb{15}{3}
\drawb{16}{3}
\drawb{17}{3}
\drawb{18}{3}
\drawb{19}{3}
\drawb{19}{2}
\drawb{18}{2}
\drawb{17}{2}
\drawb{16}{2}
\drawb{15}{2}
\drawb{14}{2}
\drawb{13}{2}
\drawb{12}{2}
\drawb{11}{2}
\drawb{10}{2}

\draw[fill=black] (9.5,3.5) circle (0.12);
\end{tikzpicture}

\caption{The path $\paths{1}{1}=\pathg\pathy{1}\pathb{1}=\dead{1}$ (for $\para=4$ and $\parb=10$). The south glue of its last tile does a mismatch. Thus, this path is a dead-end .}
\label{fig:path11}
\end{figure}

A corollary of this lemma is that $\paths{1}{1}$ is a dead-end (see Figure \ref{fig:path11}). Also, for all $2\leq j \leq i \leq \para$, let $\vu{i}{j}$ be the vector $(\parb,\series_{i}-\series_{j}+\series_{j-1})$ and since the last tile of $\paths{i}{j}$ is $(\parb,\series_{i}-\series_{j}+\series_{j-1})$ and since the type of its east east glue is $\glueg_0$ then for all $\patha$ producible by $\tas$, if $\paths{i}{j}(\pathg_0+\vu{i}{j})$ is a prefix of $\patha$ then $\patha$ can be written $\paths{i}{j}\pathaa$ where $\pathaa-\vu{i}{j}$ is a path producible by $\tas$.

\subsection{Analyzes of the prefixes.}
\label{long:sec:prefix}

Now, remark that the only tile with a glue $\glueg_0$ on its west side is $\tg_0$ then any path producible by $\tas$ begins by a tile $\tg_0$. In fact, this reasoning can be done for any glue and direction. Thus for any path $\patha$ and any $0\leq i <|\patha|-1 $, if we know the position and the tile type of $\patha_i$ and the position of $\patha_{i+1}$ then we can deduce the tile type of $\patha_{i+1}$. Thus, consider two paths $\patha$ and $\pathaa$ which are producible by $\tasr{\para}$, then they share a common prefix until they split away. They can split away only at a tile with at least three glues on its sides. From lemma \ref{lem:freeglue}, we have the following fact:

\begin{fact}
\label{lem:prefix1}
For any $\para,\parb \in \mathbb{N}$ and for any path $\patha$ producible by $\tas$, either $\patha$ is a prefix of $\pathg\pathy{\para}$ or there exists $1 \leq i \leq \para$ such that $\pathg\pathy{i}\pathb{i}_0$ is a prefix of $\patha$.
\end{fact}

There are only $\para$ different prefixes for a producible path with is large enough. Now, let's look at how these different prefixes can grow. For all $1\leq i \leq \para$, by attaching tiles to the end of $\pathg\pathy{i}\pathb{i}_0$, this pall will always grow into $\pathg\pathy{i}\pathb{i}$. For all $1\leq i \leq \para$, we define $\dead{i}$ as the path obtained by attaching red tiles to $\pathg\pathy{i}\pathb{i}$ until it is no more possible (see Figure \ref{fig:dead3}). For the special case $i=\para$, we have $\dead{\para}=\paths{\para}{\para}$ since we run out of red tiles. For the other special case $i=1$, the path $\dead{1}$ is $\paths{1}{1}$ since this path is a dead-end then it is not possible to add any red tiles (see Figure \ref{fig:path11}). For the general case $1<i<\para$, $\dead{i}$ is a dead-end since its last tile is $((\parb,1),\tilep_{\series_{i}-4})$ and its only free glue is the south one which does a mismatch with the tile $\pathg_{\parb}=((\parb,1),\tileg_{\parb})$. In all cases, there are only $i-1$ red tiles with a free glue on their east side (see Lemma \ref{lem:freeglue}). Thus if a path producible by $\tas$ is not a prefix of $\dead{i}$ then it has to use on these free glues. The previous remarks are summarized in the following fact.

\begin{figure}
\centering

\begin{tikzpicture}[x=0.25cm,y=0.25cm]

\path [draw, gray, very thin] (-4,-2) grid[step=0.25cm] (22,28);

 \draw[->] (-4,0.5) -- (22,0.5);
 \draw[->] (0.5,-2) -- (0.5,28);
 \draw[dotted,thick] (-4,26) -- (22,26);
 \draw[dotted,thick] (20,-2) -- (20,28);

\draw (-1.8,3.5) node {$\series_{1}-1$};
 \draw(0.33,3.5) -- (0.66,3.5);
\draw (-1.8,9.5) node {$\series_{2}-1$};
 \draw(0.33,9.5) -- (0.66,9.5);
\draw (-1.8,25.5) node {$\series_{3}-1$};
 \draw(0.33,25.5) -- (0.66,25.5);
\draw (10.5,-1) node {$\parb$};
\draw (19.5,-1) node {$2\parb-1$};
 \draw(19.5,0.33) -- (19.5,0.66);

 \draw[very thick,xgreen] (0.5,0.5) -| (10.5,1.5) -| (9.5,1.5);
 \draw[very thick,xorange] (9.5,1.5) |- (9.5,26);
 \draw[very thick,xblue] (9.5,25.5) -| (19.5,24.5) -| (10.5,24.5);

 \draw[very thick,xblue] (9.5,3.5) |- (10,3.5);
 \draw[very thick,xblue] (9.5,9.5) |- (10,9.5);

 \draw[very thick,xred] (10.5,24.5) -| (10.5,2);

\draws{0}{0}
\drawg{1}{0}
\drawg{2}{0}
\drawg{3}{0}
\drawg{4}{0}
\drawg{5}{0}
\drawg{6}{0}
\drawg{7}{0}
\drawg{8}{0}
\drawg{9}{0}
\drawg{10}{0}
\drawg{10}{1}
\drawg{9}{1}

\drawo{9}{2}
\drawo{9}{3}
\drawo{9}{4}
\drawo{9}{5}
\drawo{9}{6}
\drawo{9}{7}
\drawo{9}{8}
\drawo{9}{9}
\drawo{9}{10}
\drawo{9}{11}
\drawo{9}{12}
\drawo{9}{13}
\drawo{9}{14}
\drawo{9}{15}
\drawo{9}{16}
\drawo{9}{17}
\drawo{9}{18}
\drawo{9}{19}
\drawo{9}{20}
\drawo{9}{21}
\drawo{9}{22}
\drawo{9}{23}
\drawo{9}{24}
\drawo{9}{25}

\drawb{10}{25}
\drawb{11}{25}
\drawb{12}{25}
\drawb{13}{25}
\drawb{14}{25}
\drawb{15}{25}
\drawb{16}{25}
\drawb{17}{25}
\drawb{18}{25}
\drawb{19}{25}
\drawb{19}{24}
\drawb{18}{24}
\drawb{17}{24}
\drawb{16}{24}
\drawb{15}{24}
\drawb{14}{24}
\drawb{13}{24}
\drawb{12}{24}
\drawb{11}{24}
\drawb{10}{24}

\drawr{10}{23}
\drawr{10}{22}
\drawr{10}{21}
\drawr{10}{20}
\drawr{10}{19}
\drawr{10}{18}
\drawr{10}{17}
\drawr{10}{16}
\drawr{10}{15}
\drawr{10}{14}
\drawr{10}{13}
\drawr{10}{12}
\drawr{10}{11}
\drawr{10}{10}
\drawr{10}{9}
\drawr{10}{8}
\drawr{10}{7}
\drawr{10}{6}
\drawr{10}{5}
\drawr{10}{4}
\drawr{10}{3}
\drawr{10}{2}

\draw[fill=black] (9.5,3.5) circle (0.12);
\draw[fill=black] (9.5,9.5) circle (0.12);
\draw[fill=black] (9.5,25.5) circle (0.12);
\draw[fill=black] (10.5,20.5) circle (0.12);
\draw[fill=black] (10.5,10.5) circle (0.12);
\end{tikzpicture}

\caption{The path $\dead{3}$ for $\para=4$ and $\parb=10$. This path is obtained by attaching red tiles to $\pathg\pathy{3}\pathb{3}$ until a conflict occurs with $\pathg$. For all $1\leq i < \para$, the path $\dead{i}$ is a dead-end. The path $\dead{\para}$ is a special case since it is equal to $\paths{\para}{\para}$ (see Figure \ref{fig:def}) whose last tile has a free glue on its east side.}
\label{fig:dead3}
\end{figure}

\begin{fact}
\label{lem:prefix2}
For any $\para,\parb \in \mathbb{N}$ and for any path $\patha$ producible by $\tas$, there exists $1\leq i \leq \para$ such that either:
\begin{itemize}
\item $\patha$ is a prefix of $\dead{i}$;
\item there exists $2 \leq j \leq i$ such that $\paths{i}{j}(\pathg_0+\vu{i}{j})$ is a prefix of $\patha$.
\end{itemize}
Moreover, for all $1\leq i <\para$, $\dead{i}$ is a dead-end.
\end{fact}

Now, we have obtained $\para$ different prefixes which are dead-end and $(\para-1)^2$ prefixes which are not. Now, let's look at how these different prefixes can grow. Consider $2\leq j\leq i \leq \para$ and by attaching tiles to the end of $\paths{i}{j}(\pathg_0+\vu{i}{j})$, this path will always grow into $\paths{i}{j}(\pathg+\vu{i}{j})$ (this assembly is producible by $\tas$ since $\pathg+\vu{i}{j}$ is east of $\pathy{i}$ and $\pathpp{j}{i}$, north of $\pathg$ and south of $\pathb{i}$ (see Lemma \ref{lem:position})). The last tile of this path is $((2\parb-1,\series_{i}-\series_{j}+\series_{j-1}+1),\tileg_{\parb+1})$ and this path can keep growing to the north by attaching orange tiles at its end. We define the path $\deadd{i}{j}$ as the path obtained by attaching orange tiles to $\paths{i}{j}(\pathg+\vu{i}{j})$ until it is no more possible (see Figure \ref{fig:deadd}). The last tile of this path is $((2\parb-1,\series_{i}-3),\tiley_{\series_{j}-\series_{j-1}-5})$ since it not possible to add a orange tile because of a mismatch with the tile $\pathb{i}_{\parb}=((2\parb-1,\series_{i}-2),\tileb_{\parb})$. Moreover if $j>2$, this path is a dead-end because the last tile of this path has no east glue. Also, $\pathg\pathy{j}+\vu{i}{j}$ cannot be a prefix of $\pathaa$ and by applying lemma \ref{lem:prefix1} to $\pathaa-\vu{i}{j}$ we obtain that either $\patha$ is a prefix of $\deadd{i}{j}$ or there exists $1\leq i' <j$ such that $\paths{i}{j}(\pathg\pathy{i'}\pathb{i'}_0)+\vu{i}{j}$ is a prefix of $\patha$. These observations are summarized in the following fact.

\begin{figure}
\centering

\begin{minipage}[c]{.46\linewidth}

\begin{tikzpicture}[x=0.25cm,y=0.25cm]

\path [draw, gray, very thin] (-4,-2) grid[step=0.25cm] (23,28);

 \draw[->] (-4,0.5) -- (23,0.5);
 \draw[->] (0.5,-2) -- (0.5,28);
 \draw[dotted,thick] (-4,26) -- (23,26);
 \draw[dotted,thick] (21,-2) -- (21,28);

\draw (-1.8,3.5) node {$\series_{1}-1$};
 \draw(0.33,3.5) -- (0.66,3.5);
\draw (-1.8,9.5) node {$\series_{2}-1$};
 \draw(0.33,9.5) -- (0.66,9.5);
\draw (-1.8,25.5) node {$\series_{3}-1$};
 \draw(0.33,25.5) -- (0.66,25.5);
\draw (10.5,-1) node {$\parb$};
\draw (19.5,-1) node {$2\parb-1$};
 \draw(19.5,0.33) -- (19.5,0.66);

 \draw[very thick,xgreen] (0.5,0.5) -| (10.5,1.5) -| (9.5,1.5);
 \draw[very thick,xorange] (9.5,1.5) |- (9.5,26);

 \draw[very thick,xblue] (9.5,3.5) |- (10,3.5);
 \draw[very thick,xblue] (9.5,9.5) |- (10,9.5);

 \draw[very thick,xblue] (9.5,25.5) -| (19.5,24.5) -| (10.5,24.5);
 \draw[very thick,xred] (10.5,24.5) -| (10.5,10);

 \draw[very thick,xgreen] (10.5,20.5) -| (11,20.5);

 \draw[very thick,xgreen] (10.5,10.5) -| (20.5,11.5) -| (19.5,11.5);

 \draw[very thick,xorange] (19.5,11.5) |- (19.5,24);
 \draw[very thick,xblue] (19.5,13.5) |- (20,13.5);
 \draw[very thick,xblue] (19.5,19.5) |- (20,19.5);

\draws{0}{0}
\drawg{1}{0}
\drawg{2}{0}
\drawg{3}{0}
\drawg{4}{0}
\drawg{5}{0}
\drawg{6}{0}
\drawg{7}{0}
\drawg{8}{0}
\drawg{9}{0}
\drawg{10}{0}
\drawg{10}{1}
\drawg{9}{1}

\drawo{9}{2}
\drawo{9}{3}
\drawo{9}{4}
\drawo{9}{5}
\drawo{9}{6}
\drawo{9}{7}
\drawo{9}{8}
\drawo{9}{9}
\drawo{9}{10}
\drawo{9}{11}
\drawo{9}{12}
\drawo{9}{13}
\drawo{9}{14}
\drawo{9}{15}
\drawo{9}{16}
\drawo{9}{17}
\drawo{9}{18}
\drawo{9}{19}
\drawo{9}{20}
\drawo{9}{21}
\drawo{9}{22}
\drawo{9}{23}
\drawo{9}{24}
\drawo{9}{25}

\drawb{10}{25}
\drawb{11}{25}
\drawb{12}{25}
\drawb{13}{25}
\drawb{14}{25}
\drawb{15}{25}
\drawb{16}{25}
\drawb{17}{25}
\drawb{18}{25}
\drawb{19}{25}
\drawb{19}{24}
\drawb{18}{24}
\drawb{17}{24}
\drawb{16}{24}
\drawb{15}{24}
\drawb{14}{24}
\drawb{13}{24}
\drawb{12}{24}
\drawb{11}{24}
\drawb{10}{24}

\drawr{10}{23}
\drawr{10}{22}
\drawr{10}{21}
\drawr{10}{20}
\drawr{10}{19}
\drawr{10}{18}
\drawr{10}{17}
\drawr{10}{16}
\drawr{10}{15}
\drawr{10}{14}
\drawr{10}{13}
\drawr{10}{12}
\drawr{10}{11}
\drawr{10}{10}

\drawg{11}{10}
\drawg{12}{10}
\drawg{13}{10}
\drawg{14}{10}
\drawg{15}{10}
\drawg{16}{10}
\drawg{17}{10}
\drawg{18}{10}
\drawg{19}{10}
\drawg{20}{10}
\drawg{20}{11}
\drawg{19}{11}

\drawo{19}{12}
\drawo{19}{13}
\drawo{19}{14}
\drawo{19}{15}
\drawo{19}{16}
\drawo{19}{17}
\drawo{19}{18}
\drawo{19}{19}
\drawo{19}{20}
\drawo{19}{21}
\drawo{19}{22}
\drawo{19}{23}

\draw[fill=black] (9.5,3.5) circle (0.12);
\draw[fill=black] (9.5,9.5) circle (0.12);
\draw[fill=black] (9.5,25.5) circle (0.12);
\draw[fill=black] (10.5,20.5) circle (0.12);
\draw[fill=black] (10.5,10.5) circle (0.12);
\draw[fill=black] (19.5,13.5) circle (0.12);
\draw[fill=black] (19.5,19.5) circle (0.12);
\end{tikzpicture}

\end{minipage} \hfill
\begin{minipage}[c]{.46\linewidth}

\begin{tikzpicture}[x=0.25cm,y=0.25cm]

\path [draw, gray, very thin] (-4,-2) grid[step=0.25cm] (23,28);

 \draw[->] (-4,0.5) -- (23,0.5);
 \draw[->] (0.5,-2) -- (0.5,28);
 \draw[dotted,thick] (-4,26) -- (23,26);
 \draw[dotted,thick] (21,-2) -- (21,28);

\draw (-1.8,3.5) node {$\series_{1}-1$};
 \draw(0.33,3.5) -- (0.66,3.5);
\draw (-1.8,9.5) node {$\series_{2}-1$};
 \draw(0.33,9.5) -- (0.66,9.5);
\draw (-1.8,25.5) node {$\series_{3}-1$};
 \draw(0.33,25.5) -- (0.66,25.5);
\draw (10.5,-1) node {$\parb$};
\draw (19.5,-1) node {$2\parb-1$};
 \draw(19.5,0.33) -- (19.5,0.66);

 \draw[very thick,xgreen] (0.5,0.5) -| (10.5,1.5) -| (9.5,1.5);
 \draw[very thick,xorange] (9.5,1.5) |- (9.5,26);

 \draw[very thick,xblue] (9.5,3.5) |- (10,3.5);
 \draw[very thick,xblue] (9.5,9.5) |- (10,9.5);

 \draw[very thick,xblue] (9.5,25.5) -| (19.5,24.5) -| (10.5,24.5);
 \draw[very thick,xred] (10.5,24.5) -| (10.5,20);

 \draw[very thick,xgreen] (10.5,20.5) -| (20.5,21.5) -| (19.5,21.5);

 \draw[very thick,xorange] (19.5,21.5) |- (19.5,24);
 \draw[very thick,xblue] (19.5,23.5) |- (20,23.5);

\draws{0}{0}
\drawg{1}{0}
\drawg{2}{0}
\drawg{3}{0}
\drawg{4}{0}
\drawg{5}{0}
\drawg{6}{0}
\drawg{7}{0}
\drawg{8}{0}
\drawg{9}{0}
\drawg{10}{0}
\drawg{10}{1}
\drawg{9}{1}

\drawo{9}{2}
\drawo{9}{3}
\drawo{9}{4}
\drawo{9}{5}
\drawo{9}{6}
\drawo{9}{7}
\drawo{9}{8}
\drawo{9}{9}
\drawo{9}{10}
\drawo{9}{11}
\drawo{9}{12}
\drawo{9}{13}
\drawo{9}{14}
\drawo{9}{15}
\drawo{9}{16}
\drawo{9}{17}
\drawo{9}{18}
\drawo{9}{19}
\drawo{9}{20}
\drawo{9}{21}
\drawo{9}{22}
\drawo{9}{23}
\drawo{9}{24}
\drawo{9}{25}

\drawb{10}{25}
\drawb{11}{25}
\drawb{12}{25}
\drawb{13}{25}
\drawb{14}{25}
\drawb{15}{25}
\drawb{16}{25}
\drawb{17}{25}
\drawb{18}{25}
\drawb{19}{25}
\drawb{19}{24}
\drawb{18}{24}
\drawb{17}{24}
\drawb{16}{24}
\drawb{15}{24}
\drawb{14}{24}
\drawb{13}{24}
\drawb{12}{24}
\drawb{11}{24}
\drawb{10}{24}

\drawr{10}{23}
\drawr{10}{22}
\drawr{10}{21}
\drawr{10}{20}

\drawg{11}{20}
\drawg{12}{20}
\drawg{13}{20}
\drawg{14}{20}
\drawg{15}{20}
\drawg{16}{20}
\drawg{17}{20}
\drawg{18}{20}
\drawg{19}{20}
\drawg{20}{20}
\drawg{20}{21}
\drawg{19}{21}

\drawo{19}{22}
\drawo{19}{23}

\draw[fill=black] (9.5,3.5) circle (0.12);
\draw[fill=black] (9.5,9.5) circle (0.12);
\draw[fill=black] (9.5,25.5) circle (0.12);
\draw[fill=black] (10.5,20.5) circle (0.12);
\draw[fill=black] (19.5,23.5) circle (0.12);
\end{tikzpicture}

\end{minipage} \hfill

\caption{The path $\deadd{3}{3}$ (on the left) and $\deadd{3}{2}$ (on the right) for $\para=4$ and $\parb=10$. These paths are obtained by attaching blue and orange tiles to $\paths{3}{3}$ and $\paths{3}{2}$ until a conflict occurs with $\pathb{3}$. For all $2\leq j\leq  i \leq \para$, the path $\deadd{i}{j}$ is a dead-end if $j \neq 2$.}
\label{fig:deadd}
\end{figure}

\begin{fact}
\label{lem:prefix3}
For any $\para,\parb \in \mathbb{N}$ and for any path $\patha$ producible by $\tas$, there exists $1 \leq i \leq \para$ such that either:
\begin{itemize}
\item $\patha$ is a prefix of $\dead{i}$;
\item there exists $2\leq j \leq i$ such that either:
\begin{itemize}
\item $\patha$ is prefix of $\deadd{i}{j}$;
\item there exists $1\leq i' <j$ such that $\paths{i}{j}(\pathg\pathy{i'}\pathb{i'}_0)+\vu{i}{j}$ is a prefix of $\patha$.
\end{itemize}
\end{itemize}
Moreover, for all $2< j \leq i \leq \para$, $\deadd{i}{j}$ is a dead-end.
\end{fact}

\begin{figure}
\centering

\begin{tikzpicture}[x=0.25cm,y=0.25cm]

\path [draw, gray, very thin] (-4,-2) grid[step=0.25cm] (52,70);

 \draw[->] (-4,0.5) -- (52,0.5);
 \draw[->] (0.5,-2) -- (0.5,70);
 \draw[dotted,thick] (-4,68) -- (52,68);
 \draw[dotted,thick] (30,-2) -- (30,70);

\draw (-1.8,3.5) node {$\series_{1}-1$};
 \draw(0.33,3.5) -- (0.66,3.5);
\draw (-1.8,9.5) node {$\series_{2}-1$};
 \draw(0.33,9.5) -- (0.66,9.5);
\draw (-1.8,25.5) node {$\series_{3}-1$};
 \draw(0.33,25.5) -- (0.66,25.5);
\draw (-1.8,67.5) node {$\series_{4}-1$};
 \draw(0.33,67.5) -- (0.66,67.5);
\draw (10.5,-1) node {$\parb$};
\draw (19.5,-1) node {$2\parb-1$};
 \draw(19.5,0.33) -- (19.5,0.66);
\draw (29.5,-1) node {$3\parb-1$};
 \draw(29.5,0.33) -- (29.5,0.66);
\draw (39.5,-1) node {$4\parb-1$};
 \draw(39.5,0.33) -- (39.5,0.66);
\draw (49.5,-1) node {$5\parb-1$};
 \draw(49.5,0.33) -- (49.5,0.66);

 \draw[very thick,xgreen] (0.5,0.5) -| (10.5,1.5) -| (9.5,1.5);
 \draw[very thick,xorange] (9.5,1.5) |- (9.5,67.5);
 \draw[very thick,xblue] (9.5,3.5) |- (10,3.5);
 \draw[very thick,xblue] (9.5,9.5) |- (10,9.5);
 \draw[very thick,xblue] (9.5,25.5) -| (10,25.5);
 \draw[very thick,xblue] (9.5,67.5) -| (19.5,66.5) -| (10.5,66.5);
 \draw[very thick,xred] (10.5,66.5) -| (10.5,52);
 \draw[very thick,xgreen] (10.5,62.5) -| (11,62.5);
 \draw[very thick,xgreen] (10.5,52.5) -| (20.5,53.5) -| (19.5,53.5);
 \draw[very thick,xorange] (19.5,53.5) |- (19.5,62);
 \draw[very thick,xblue] (19.5,55.5) |- (20,55.5);
 \draw[very thick,xblue] (19.5,61.5) -| (29.5,60.5) -| (20.5,60.5);
 \draw[very thick,xred] (20.5,60.5) -| (20.5,54);
 \draw[very thick,xgreen] (20.5,56.5) -| (21,56.5);

\draws{0}{0}
\drawg{1}{0}
\drawg{2}{0}
\drawg{3}{0}
\drawg{4}{0}
\drawg{5}{0}
\drawg{6}{0}
\drawg{7}{0}
\drawg{8}{0}
\drawg{9}{0}
\drawg{10}{0}
\drawg{10}{1}
\drawg{9}{1}

\drawo{9}{2}
\drawo{9}{3}
\drawo{9}{4}
\drawo{9}{5}
\drawo{9}{6}
\drawo{9}{7}
\drawo{9}{8}
\drawo{9}{9}
\drawo{9}{10}
\drawo{9}{11}
\drawo{9}{12}
\drawo{9}{13}
\drawo{9}{14}
\drawo{9}{15}
\drawo{9}{16}
\drawo{9}{17}
\drawo{9}{18}
\drawo{9}{19}
\drawo{9}{20}
\drawo{9}{21}
\drawo{9}{22}
\drawo{9}{23}
\drawo{9}{24}
\drawo{9}{25}
\drawo{9}{26}
\drawo{9}{27}
\drawo{9}{28}
\drawo{9}{29}
\drawo{9}{30}
\drawo{9}{31}
\drawo{9}{32}
\drawo{9}{33}
\drawo{9}{34}
\drawo{9}{35}
\drawo{9}{36}
\drawo{9}{37}
\drawo{9}{38}
\drawo{9}{39}
\drawo{9}{40}
\drawo{9}{41}
\drawo{9}{42}
\drawo{9}{43}
\drawo{9}{44}
\drawo{9}{45}
\drawo{9}{46}
\drawo{9}{47}
\drawo{9}{48}
\drawo{9}{49}
\drawo{9}{50}
\drawo{9}{51}
\drawo{9}{52}
\drawo{9}{53}
\drawo{9}{54}
\drawo{9}{55}
\drawo{9}{56}
\drawo{9}{57}
\drawo{9}{58}
\drawo{9}{59}
\drawo{9}{60}
\drawo{9}{61}
\drawo{9}{62}
\drawo{9}{63}
\drawo{9}{64}
\drawo{9}{65}
\drawo{9}{66}
\drawo{9}{67}

\drawb{10}{67}
\drawb{11}{67}
\drawb{12}{67}
\drawb{13}{67}
\drawb{14}{67}
\drawb{15}{67}
\drawb{16}{67}
\drawb{17}{67}
\drawb{18}{67}
\drawb{19}{67}
\drawb{19}{66}
\drawb{18}{66}
\drawb{17}{66}
\drawb{16}{66}
\drawb{15}{66}
\drawb{14}{66}
\drawb{13}{66}
\drawb{12}{66}
\drawb{11}{66}
\drawb{10}{66}

\drawr{10}{65}
\drawr{10}{64}
\drawr{10}{63}
\drawr{10}{62}
\drawr{10}{61}
\drawr{10}{60}
\drawr{10}{59}
\drawr{10}{58}
\drawr{10}{57}
\drawr{10}{56}
\drawr{10}{55}
\drawr{10}{54}
\drawr{10}{53}
\drawr{10}{52}

\drawg{11}{52}
\drawg{12}{52}
\drawg{13}{52}
\drawg{14}{52}
\drawg{15}{52}
\drawg{16}{52}
\drawg{17}{52}
\drawg{18}{52}
\drawg{19}{52}
\drawg{20}{52}
\drawg{20}{53}
\drawg{19}{53}

\drawo{19}{54}
\drawo{19}{55}
\drawo{19}{56}
\drawo{19}{57}
\drawo{19}{58}
\drawo{19}{59}
\drawo{19}{60}
\drawo{19}{61}

\drawb{20}{61}
\drawb{21}{61}
\drawb{22}{61}
\drawb{23}{61}
\drawb{24}{61}
\drawb{25}{61}
\drawb{26}{61}
\drawb{27}{61}
\drawb{28}{61}
\drawb{29}{61}
\drawb{29}{60}
\drawb{28}{60}
\drawb{27}{60}
\drawb{26}{60}
\drawb{25}{60}
\drawb{24}{60}
\drawb{23}{60}
\drawb{22}{60}
\drawb{21}{60}
\drawb{20}{60}

\drawr{20}{59}
\drawr{20}{58}
\drawr{20}{57}
\drawr{20}{56}
\drawr{20}{55}
\drawr{20}{54}

\draw[fill=black] (9.5,3.5) circle (0.12);
\draw[fill=black] (9.5,9.5) circle (0.12);
\draw[fill=black] (9.5,25.5) circle (0.12);
\draw[fill=black] (9.5,67.5) circle (0.12);
\draw[fill=black] (10.5,62.5) circle (0.12);
\draw[fill=black] (10.5,52.5) circle (0.12);
\draw[fill=black] (19.5,55.5) circle (0.12);
\draw[fill=black] (19.5,61.5) circle (0.12);
\draw[fill=black] (20.5,56.5) circle (0.12);

\draw[->,>=stealth,thick](1.5,0.5)--(11.5, 52.5) node [midway,above,sloped]{$\vu{4}{3}$};

\end{tikzpicture}

\caption{This path is the concatenation of $\paths{4}{3}$ and $\dead{2}$ up to some translation. It is producible by $\tas$ (for $\para=4$ and $\parb=10$) and it is a dead-end this its suffix is $\dead{2}$ up to some translation.}
\label{fig:def}
\end{figure}

\subsection{Analysis of the tile assembly system.}
\label{long:sec:analyzes}

Consider a path $\patha$ producible by $\tas$ then Fact \ref{lem:prefix3} gives us an hint to the structure of $\patha$. Indeed, this path is made of several paths $\paths{i_1}{j_1}$, $\paths{i_2}{j_2}$, $\paths{i_3}{j_3}$, $\ldots$ which are glued together up to some translations (see Figure \ref{}). The index $i$ (resp. $j$) represents which east glue of an orange (resp. red) tile is used. Moreover, Fact \ref{lem:prefix3} also implies that $i_1<i_2<i_3 < \ldots$ and thus eventually one the three dead-end $\paths{1}{1}$, $\dead{i}$, $\deadd{i}{j}$ (for some $2\leq j \leq i \leq \para$) will grow up to some translation and the path will become a dead-end too. We use this sketch of proof to show that the width and height of any producible path is bounded.

\begin{lemma}
\label{lem:prefixbound}
For any $\para,\parb \in \mathbb{N}$ and for any path $\patha$ producible by $\tas$, if there exists $1\leq i \leq \para$ such that $\pathg\pathy{i}\pathb{i}_0$ is a prefix of $\patha$ then $\patha_0$ is its southernmost and westernmost tile and its height is $\series_{i}-1$ and its width is bounded by $(i+1)\parb-1$.
\end{lemma}

\begin{proof}
This proof is done by recurrence on $i$. For $i=1$, if $\pathg\pathy{1}\pathb{1}_0$ is a prefix of $\patha$ then by Fact~\ref{lem:prefix2}, $\patha$ is a prefix of $\paths{1}{1}$ and thus its height is $\series_1-1=3$ and its width is bounded by $2\parb-1$ (see Lemma \ref{lem:position}). Then the initialization of the recurrence is true. Now, consider $1\leq i \leq \para-1$ such that the hypothesis of recurrence is true for all $1\leq i' \leq i$ and consider a path $\patha$ producible by $\tas$ whose prefix is $\pathg\pathy{i+1}\pathb{i+1}_0$. If $\patha$ is a prefix of $\dead{i+1}$ or $\deadd{i+1}{j}$ (for some $1\leq j \leq i+1$) then its height is $\series_{i+1}-1$ and its width is bounded by $2\parb$ (see lemma \ref{lem:position}). In this case, the hypothesis of recurrence is true for $i+1$. Otherwise by fact \ref{lem:prefix3}, there exists $1\leq i'< j \leq i+1$ such that $\patha$ can be written as $\paths{i+1}{j}\pathaa$ and where $(\pathg\pathy{i'}\pathb{i'}_0)+\vu{i+1}{j}$ is a prefix of $\pathaa$ then by recurrence the height of $\pathaa-\vu{i+1}{j}$ is $\series_{i'}-1$ and its width is bounded by $(i'+1)\parb-1$. Thus the height of $\patha$ is bounded by $\series_{i+1}$ and its width is bounded by $(i+2)\parb-1$. The recurrence is true for $i+1$.
\end{proof}

This result and fact \ref{lem:prefix1} imply that we are working on an area of the plane delimited by $0\leq x \leq (\para+1)\parb-1$ and $0\leq y \leq \series_{\parb+1}$ and that for any path $\patha$ producible by $\tas$, its first tile is the westernmost and southernmost one.

\begin{corollary}
\label{cor:bounded}
For any $\para,\parb \in \mathbb{N}$ and for any terminal assembly $\alpha \in \termasm{\tas}$, the height of $\alpha$ is bounded by $\series_{\parb+1}$ and its width is bounded by $(\para+1)\parb-1$.
\end{corollary}

Now, we aim to assemble the largest path possible. Note that, for all $1\leq i \leq \para$, the last tile of the path $\paths{i}{i}$ is at distance $\parb$ to the east of its first tile (this property is due to the definition of the green tiles). Thus, if we manage to glue $\para$ paths together, we will obtain a path whose width is $(\para+1)\parb-1$. This is the path we are looking for. Formally, we define the path $\pathl^1$ as $\paths{1}{1}$ and for all $2\leq i \leq \para$, we define $\pathl^{i}$ as $\paths{i}{i}(\pathl^{i-1}+\vu{i}{i})$. We remind that $\vu{i}{i}=(\parb,\series_{i-1})$. Now, we show that this sequence of tiles is a path producible by $\tas$.

\begin{figure}
\centering

\begin{tikzpicture}[x=0.25cm,y=0.25cm]

\path [draw, gray, very thin] (-4,-2) grid[step=0.25cm] (52,70);

 \draw[->] (-4,0.5) -- (52,0.5);
 \draw[->] (0.5,-2) -- (0.5,70);
 \draw[dotted,thick] (-4,68) -- (52,68);
 \draw[dotted,thick] (20,-2) -- (20,70);

\draw (-1.8,3.5) node {$\series_{1}-1$};
 \draw(0.33,3.5) -- (0.66,3.5);
\draw (-1.8,9.5) node {$\series_{2}-1$};
 \draw(0.33,9.5) -- (0.66,9.5);
\draw (-1.8,25.5) node {$\series_{3}-1$};
 \draw(0.33,25.5) -- (0.66,25.5);
\draw (-1.8,67.5) node {$\series_{4}-1$};
 \draw(0.33,67.5) -- (0.66,67.5);
\draw (10.5,-1) node {$\parb$};
\draw (19.5,-1) node {$2\parb-1$};
 \draw(19.5,0.33) -- (19.5,0.66);
\draw (29.5,-1) node {$3\parb-1$};
 \draw(29.5,0.33) -- (29.5,0.66);
\draw (39.5,-1) node {$4\parb-1$};
 \draw(39.5,0.33) -- (39.5,0.66);
\draw (49.5,-1) node {$5\parb-1$};
 \draw(49.5,0.33) -- (49.5,0.66);

 \draw[very thick,xgreen] (0.5,0.5) -| (10.5,1.5) -| (9.5,1.5);
 \draw[very thick,xorange] (9.5,1.5) |- (9.5,67.5);
 \draw[very thick,xblue] (9.5,3.5) |- (10,3.5);
 \draw[very thick,xblue] (9.5,9.5) |- (10,9.5);
 \draw[very thick,xblue] (9.5,25.5) -| (10,25.5);
 \draw[very thick,xblue] (9.5,67.5) -| (19.5,66.5) -| (10.5,66.5);
 \draw[very thick,xred] (10.5,66.5) -| (10.5,26.5);
 \draw[very thick,xgreen] (10.5,26.5) -| (11,26.5);
 \draw[very thick,xgreen] (10.5,52.5) -| (11,52.5);
 \draw[very thick,xgreen] (10.5,62.5) -| (11,62.5);

\draws{0}{0}
\drawg{1}{0}
\drawg{2}{0}
\drawg{3}{0}
\drawg{4}{0}
\drawg{5}{0}
\drawg{6}{0}
\drawg{7}{0}
\drawg{8}{0}
\drawg{9}{0}
\drawg{10}{0}
\drawg{10}{1}
\drawg{9}{1}

\drawo{9}{2}
\drawo{9}{3}
\drawo{9}{4}
\drawo{9}{5}
\drawo{9}{6}
\drawo{9}{7}
\drawo{9}{8}
\drawo{9}{9}
\drawo{9}{10}
\drawo{9}{11}
\drawo{9}{12}
\drawo{9}{13}
\drawo{9}{14}
\drawo{9}{15}
\drawo{9}{16}
\drawo{9}{17}
\drawo{9}{18}
\drawo{9}{19}
\drawo{9}{20}
\drawo{9}{21}
\drawo{9}{22}
\drawo{9}{23}
\drawo{9}{24}
\drawo{9}{25}
\drawo{9}{26}
\drawo{9}{27}
\drawo{9}{28}
\drawo{9}{29}
\drawo{9}{30}
\drawo{9}{31}
\drawo{9}{32}
\drawo{9}{33}
\drawo{9}{34}
\drawo{9}{35}
\drawo{9}{36}
\drawo{9}{37}
\drawo{9}{38}
\drawo{9}{39}
\drawo{9}{40}
\drawo{9}{41}
\drawo{9}{42}
\drawo{9}{43}
\drawo{9}{44}
\drawo{9}{45}
\drawo{9}{46}
\drawo{9}{47}
\drawo{9}{48}
\drawo{9}{49}
\drawo{9}{50}
\drawo{9}{51}
\drawo{9}{52}
\drawo{9}{53}
\drawo{9}{54}
\drawo{9}{55}
\drawo{9}{56}
\drawo{9}{57}
\drawo{9}{58}
\drawo{9}{59}
\drawo{9}{60}
\drawo{9}{61}
\drawo{9}{62}
\drawo{9}{63}
\drawo{9}{64}
\drawo{9}{65}
\drawo{9}{66}
\drawo{9}{67}

\drawb{10}{67}
\drawb{11}{67}
\drawb{12}{67}
\drawb{13}{67}
\drawb{14}{67}
\drawb{15}{67}
\drawb{16}{67}
\drawb{17}{67}
\drawb{18}{67}
\drawb{19}{67}
\drawb{19}{66}
\drawb{18}{66}
\drawb{17}{66}
\drawb{16}{66}
\drawb{15}{66}
\drawb{14}{66}
\drawb{13}{66}
\drawb{12}{66}
\drawb{11}{66}
\drawb{10}{66}

\drawr{10}{65}
\drawr{10}{64}
\drawr{10}{63}
\drawr{10}{62}
\drawr{10}{61}
\drawr{10}{60}
\drawr{10}{59}
\drawr{10}{58}
\drawr{10}{57}
\drawr{10}{56}
\drawr{10}{55}
\drawr{10}{54}
\drawr{10}{53}
\drawr{10}{52}
\drawr{10}{51}
\drawr{10}{50}
\drawr{10}{49}
\drawr{10}{48}
\drawr{10}{47}
\drawr{10}{46}
\drawr{10}{45}
\drawr{10}{44}
\drawr{10}{43}
\drawr{10}{42}
\drawr{10}{41}
\drawr{10}{40}
\drawr{10}{39}
\drawr{10}{38}
\drawr{10}{37}
\drawr{10}{36}
\drawr{10}{35}
\drawr{10}{34}
\drawr{10}{33}
\drawr{10}{32}
\drawr{10}{31}
\drawr{10}{30}
\drawr{10}{29}
\drawr{10}{28}
\drawr{10}{27}
\drawr{10}{26}

\draw[fill=black] (9.5,3.5) circle (0.12);
\draw[fill=black] (9.5,9.5) circle (0.12);
\draw[fill=black] (9.5,25.5) circle (0.12);
\draw[fill=black] (9.5,67.5) circle (0.12);
\draw[fill=black] (10.5,62.5) circle (0.12);
\draw[fill=black] (10.5,52.5) circle (0.12);
\draw[fill=black] (10.5,26.5) circle (0.12);
\end{tikzpicture}

\caption{In our examples, we consider $\para=4$ and $\parb=10$. The seed (in white) is at position $(0,0)$ and we represent the path $\paths{4}{4}=\pathg\pathy{4}\pathb{4}\pathpp{4}{4}$. This path is producible by $\tas$ and this figure contains exactly one occurrence of each tile type of $\glueset$. The height of $\paths{4}{4}$ is $\series_{\para}-1$ and its width is $2\parb-1$. The tiles of $\pathy{4}$ and $\pathpp{4}{4}$ with a glue on their east side are marked by a black dot.}
\label{fig:def}
\end{figure}

\begin{lemma}
\label{lem:longpath}
For all $1 \leq i \leq \para$, $\pathl^i$ is a path producible by $\tas$ of width $(i+1)\parb-1$.
\end{lemma}

\begin{proof}
This proof is done by recurrence on $i$. By lemma \ref{lem:producible}, $\paths{1}{1}$ is producible by $\tas$ and the position of the tile $\pathb{1}_{\parb-1}$ is $(2\parb-1,3)$ Thus the initialization of the recurrence is true. Now, consider $1 \leq i < \para$ such that the recurrence is true for $i$. By definition, $\pathg\pathy{i}\pathb{i}_0$ is a prefix of~$\pathl^{i}$. Then by lemma \ref{lem:prefixbound}, for any position $(x,y) \in \mathbb{N}$ occupied by a tile of $\pathl^{\step}+\vu{i+1}{i+1}$, we have $\parb+1\leq x \leq (\step+2)\parb-1$ and $\series_\step \leq y \leq 2\series_{\step}-1$. By lemma \ref{lem:position}, $\pathl^{\step}+\vu{i+1}{i+1}$ is west of $\pathy{\step+1}$ and $\pathp{\step+1}$, north of $\pathg$ and south of $\pathb{i+1}$ and thus $\pathl^{i}+\vu{i}{i}$ neither intersects $\paths{\step+1}{\step+1}$ nor the seed. By recurrence $\pathl^{i}$ is a path producible by $\tas$ and thus $\paths{i+1}{i+1}(\pathl^{i}+\vu{i}{i})$ is also a path producible by $\tas$. Finally, the width of $\paths{i}{i}$ is $(i+1)\parb-1$ by recurrence and then the width of $\pathl^{i+1}$ is $(i+2)\parb-1$.
\end{proof}

\begin{lemma}
\label{lem:alwayshere}
For any $\para,\parb \in \mathbb{N}$ and for any terminal assembly $\alpha \in \termasm{\tas}$, $\pathl^{\para}$ is a subassembly of $\alpha$.
\end{lemma}

\begin{proof}
Consider the following hypothesis of recurrence for $1\leq i \leq \para$: "Consider a path $\patha$ producible by $\tas$, if there exists  $1\leq i' \leq i$ such that  $\pathg\pathy{i'}\pathb{i'}_0$ is prefix a $\patha$ then there is no conflict between $\pathl^{i}$ and $\patha$". By fact \ref{lem:prefix2}, if $\pathg\pathy{1}\pathb{1}_0$ is a prefix of $\patha$ then $\patha$ is a prefix of $\paths{1}{1}=\pathl^{1}$ and the recurrence is true for $i=1$. Now suppose that the recurrence is true for $1\leq i < \para$. 
By definition, $\pathg\pathy{i}\pathb{i}_0$ is a prefix of~$\pathl^{i}$. Then by lemma \ref{lem:prefixbound}, for any position $(x,y) \in \mathbb{N}$ occupied by a tile of $\pathl^{i}+\vu{i+1}{i+1}$, we have $\series_\step \leq y \leq 2\series_{\step}-1$. Consider a path $\patha$ producible by $\tas$ such that $\pathg\pathy{i'}\pathb{i'}_0$ is prefix of $\patha$ for some $1\leq i' \leq i$. To prove the recurrence for $i+1$, we have to study five cases which are illustrated in Figure \ref{}. Firstly, if $i'<i$ then by lemma \ref{lem:prefixbound}, $\patha$ is south of $\pathl^{i}+\vu{i+1}{i+1}$. Also, $\patha$ is south of the suffix of $\paths{i+1}{i+1}$ obtained by removing the prefix $\pathg\pathy{i'}$ to $\paths{i+1}{i+1}$. Thus, there is no conflict between $\patha$ and $\pathl^{i+1}$ in this case. Secondly, suppose that $\patha$ is a prefix of $\dead{i+1}$ then $\paths{i+1}{i+1}$ is a prefix of both $\pathl^{i+1}$ and $\dead{i+1}$. Also, $\pathl^{i}+\vu{i+1}{i+1}$ is north of the suffix of $\dead{i+1}$ obtained by removing the prefix $\paths{i+1}{i+1}$ to $\dead{i+1}$. Thus, there is no conflict between $\patha$ and $\pathl^{i+1}$ in this case. 
Thirdly, if there exists $1\leq j < i+1$ such that $\paths{i}{j}(\pathg_0+\vu{i}{j})$ is a prefix of $\patha$. Then, $\paths{i+1}{j}$ is a prefix of both $\pathl^{i+1}$ and $\patha$. Moreover, $\patha$ can be written as $\paths{i+1}{j}\pathaa$ where $\pathaa-\vu{i}{j}$ is producible by $\tas$. Thus by lemma \ref{lem:prefixbound}, $\pathaa_0$ is the southernmost tile of $\pathaa$. Then for any position $(x,y) \in \mathbb{N}^2$ occupied by a tile of $\pathaa$, we have $y\geq \series_{i+1}-\series_{j}+\series_{j-1} \geq \series_{i+1}-\series_{i}+\series_{i-1}\geq 2\series_i$. Then $\pathaa$ is north of $\pathl^{i}+\vu{i+1}{i+1}$. Also, $\pathaa$ is north of the suffix of $\paths{i+1}{i+1}$ obtained by removing the prefix $\paths{i+1}{j}$ to $\paths{i+1}{i+1}$. Thus, there is no conflict between $\patha$ and $\pathl^{i+1}$ in this case.
Fourthly, suppose that $\patha$ is prefix of $\deadd{i+1}{i+1}$ then $\paths{i+1}{i+1}(\pathg\pathy{i})+\vu{i+1}{i+1}$ is a prefix of both $\deadd{i+1}{i+1}$ and $\pathl^{i+1}$. Also, $\pathl^{i}+\vu{i+1}{i+1}$ is south of the suffix of $\deadd{i+1}{i+1}$ obtained by removing the prefix $\paths{i+1}{i+1}(\pathg\pathy{i})+\vu{i+1}{i+1}$ to $\deadd{i+1}{i+1}$. Thus, there is no conflict between $\patha$ and $\pathl^{i+1}$ in this case. Fifthly, if $\patha$ does no match the previous cases, then by fact \ref{lem:prefix3}, there exists $1\leq i' <i+1$ such that $\patha$ can be written as $\paths{i+1}{i+1}\pathaa$ where $(\pathg\pathy{i'})+\vu{i+1}{i+1}$ is a prefix of $\pathaa$. Since $i<i+1$ then by recurrence, there is no conflict between $\pathaa-\vu{i+1}{i+1}$ and $\pathl^{i}$. Since $\paths{i+1}{i+1}$ is prefix of $\patha$ and $\pathl^{i+1}$ then there is no conflict between $\patha$ and $\pathl^{i+1}$ in this final case. Thus the recurrence is true and for any path $\patha$ producible by $\tas$, if there exists  $1\leq i \leq \para$ such that  $\pathg\pathy{i}\pathb{i}_0$ is prefix a $\patha$ then there is no conflict between $\pathl^{\para}$ and $\patha$. Since $\pathg\pathy{\para}$ is a prefix of $\pathl^{\para}$ then by fact \ref{lem:prefix1}, the lemma is true. 
\end{proof}

\begin{figure}
\centering

\begin{tikzpicture}[x=0.25cm,y=0.25cm]

\path [draw, gray, very thin] (-4,-2) grid[step=0.25cm] (52,28);

 \draw[->] (-4,0.5) -- (52,0.5);
 \draw[->] (0.5,-2) -- (0.5,28);
 \draw[dotted,thick] (-4,26) -- (52,26);
 \draw[dotted,thick] (21,-2) -- (21,28);

 \draw[dotted,thick] (-4,10) -- (52,10);
 \draw[dotted,thick] (-4,20) -- (52,20);

\node[rectangle,draw,thick] (g) at (20.5,8.5)  {\small the suffix of $\dead{3}$ is to the south.};
\node[rectangle,draw,thick] (g) at (26.75,3.5)  {\small path growing from $\pathg\pathy{i}$ with $i\leq2$ are to the south.};
\node[rectangle,draw,thick] (g) at (28.25,23.5)  {\small path growing from $\paths{3}{j}$ with $j\leq 2$ are to the north.};
\node[rectangle,draw,thick] (g) at (27,21.25)  {\small the suffix of $\deadd{3}{3}$ is to the north.};
\node[rectangle,draw,thick,text width=7cm] (g) at (35.25,15.5)  {\small path growing from $\paths{3}{3}(\pathg\pathy{i'})+\vu{3}{3}$ with $i'\leq 2$ are dealt with the hypothesis of recurrence.};
\draw (-1.8,3.5) node {$\series_{1}-1$};
 \draw(0.33,3.5) -- (0.66,3.5);
\draw (-1.8,9.5) node {$\series_{2}-1$};
 \draw(0.33,19.5) -- (0.66,19.5);
\draw (-2.2,19.5) node {$2\series_{2}-1$};
 \draw(0.33,9.5) -- (0.66,9.5);
\draw (-1.8,25.5) node {$\series_{3}-1$};
 \draw(0.33,51.5) -- (0.66,51.5);
\draw (10.5,-1) node {$\parb$};
\draw (19.5,-1) node {$2\parb-1$};
 \draw(19.5,0.33) -- (19.5,0.66);
\draw (29.5,-1) node {$3\parb-1$};
 \draw(29.5,0.33) -- (29.5,0.66);
\draw (39.5,-1) node {$4\parb-1$};
 \draw(39.5,0.33) -- (39.5,0.66);
\draw (49.5,-1) node {$5\parb-1$};
 \draw(49.5,0.33) -- (49.5,0.66);

 \draw[very thick,xgreen] (0.5,0.5) -| (10.5,1.5) -| (9.5,1.5);
 \draw[very thick,xorange] (9.5,1.5) |- (9.5,26);
 \draw[very thick,xblue] (9.5,3.5) |- (10,3.5);
 \draw[very thick,xblue] (9.5,9.5) |- (10,9.5);
 \draw[very thick,xblue] (9.5,25.5) -| (19.5,24.5) -| (10.5,24.5);
 \draw[very thick,xred] (10.5,24.5) -| (10.5,10);
 \draw[very thick,xgreen] (10.5,20.5) -| (11,20.5);
 
 \draw[very thick,xgreen] (10.5,10.5) -| (20.5,11.5) -| (19.5,11.5);
 \draw[very thick,xorange] (19.5,11.5) |- (19.5,20);
 \draw[very thick,xblue] (19.5,13.5) |- (20,13.5);
 \draw[very thick,xblue] (19.5,19.5) |- (20,19.5);

\draws{0}{0}
\drawg{1}{0}
\drawg{2}{0}
\drawg{3}{0}
\drawg{4}{0}
\drawg{5}{0}
\drawg{6}{0}
\drawg{7}{0}
\drawg{8}{0}
\drawg{9}{0}
\drawg{10}{0}
\drawg{10}{1}
\drawg{9}{1}

\drawo{9}{2}
\drawo{9}{3}
\drawo{9}{4}
\drawo{9}{5}
\drawo{9}{6}
\drawo{9}{7}
\drawo{9}{8}
\drawo{9}{9}
\drawo{9}{10}
\drawo{9}{11}
\drawo{9}{12}
\drawo{9}{13}
\drawo{9}{14}
\drawo{9}{15}
\drawo{9}{16}
\drawo{9}{17}
\drawo{9}{18}
\drawo{9}{19}
\drawo{9}{20}
\drawo{9}{21}
\drawo{9}{22}
\drawo{9}{23}
\drawo{9}{24}
\drawo{9}{25}

\drawb{10}{25}
\drawb{11}{25}
\drawb{12}{25}
\drawb{13}{25}
\drawb{14}{25}
\drawb{15}{25}
\drawb{16}{25}
\drawb{17}{25}
\drawb{18}{25}
\drawb{19}{25}
\drawb{19}{24}
\drawb{18}{24}
\drawb{17}{24}
\drawb{16}{24}
\drawb{15}{24}
\drawb{14}{24}
\drawb{13}{24}
\drawb{12}{24}
\drawb{11}{24}
\drawb{10}{24}

\drawr{10}{23}
\drawr{10}{22}
\drawr{10}{21}
\drawr{10}{20}
\drawr{10}{19}
\drawr{10}{18}
\drawr{10}{17}
\drawr{10}{16}
\drawr{10}{15}
\drawr{10}{14}
\drawr{10}{13}
\drawr{10}{12}
\drawr{10}{11}
\drawr{10}{10}


\drawg{11}{10}
\drawg{12}{10}
\drawg{13}{10}
\drawg{14}{10}
\drawg{15}{10}
\drawg{16}{10}
\drawg{17}{10}
\drawg{18}{10}
\drawg{19}{10}
\drawg{20}{10}
\drawg{20}{11}
\drawg{19}{11}

\drawo{19}{12}
\drawo{19}{13}
\drawo{19}{14}
\drawo{19}{15}
\drawo{19}{16}
\drawo{19}{17}
\drawo{19}{18}
\drawo{19}{19}

\draw[fill=black] (9.5,3.5) circle (0.12);
\draw[fill=black] (9.5,9.5) circle (0.12);
\draw[fill=black] (9.5,25.5) circle (0.12);
\draw[fill=black] (10.5,20.5) circle (0.12);
\draw[fill=black] (10.5,10.5) circle (0.12);

\draw[fill=black] (19.5,13.5) circle (0.12);
\draw[fill=black] (19.5,19.5) circle (0.12);

\end{tikzpicture}

\caption{Graphical representation of the step of recurrence of the proof of Lemma \ref{lem:alwayshere} for $\para=4$,$\parb=10$ and $i=3$. We represent here a prefix of $\pathl^{3}$, note that the remaining suffix of $\pathl^{3}$ is in the stripe defined by $\series_2\leq y \leq 2\series_2-1$. If a path producible by $\tas$ creates a conflit with $\pathl^{3}$ then it has to fork from this prefix by one of its free glue. We represent the five possible cases and give the main argument to deal with each case. }
\label{fig:proof:stable}
\end{figure}

\subsection{Conclusion of the proof}
\label{sec:conc}

Now, we obtain our main result by setting the parameters $\para$ and $\parb$ correctly and by combining the main results of the previous section (corollary \ref{cor:bounded} and lemmas \ref{lem:longpath} and \ref{lem:alwayshere}).

\begin{lemma}
\label{lem:pathlog}
For all $\para \in \mathbb{N}$, there exists a path $\patha$ producible by a tile assembly system $(\settiles,\sigma,1)$ such that:
\begin{itemize}
\item the size of $\settiles$ is $8\series_\para-\series_{\para-1}-1$;
\item the width of $\patha$ is $\Theta(|\settiles| \log_3 (|\settiles|))$;
\item for any terminal assembly $\alpha \in \termasm{\mathcal T}$, the width and height of $\alpha$ are bounded by $\Theta(|\settiles| \log_3 (|\settiles|))$ and $\asm{\patha}$ is a subassembly of $\alpha$.
\end{itemize} 
\end{lemma}

\begin{proof}
Consider $\para,\parb \in \mathbb{N}$ and the tile assembly system $\tas$, its set of tile types is made of:
\begin{itemize}
\item $1$ tiles types for the seed;
\item $\parb+2$ green tiles types;
\item $\series_\para-2$ orange tile types;
\item $2\parb$ blue tile types;
\item $\series_\para-\series_{\para-1}-2$ red tile types.
\end{itemize}
Thus, we have $|\glueset|=3\parb+2\series_\para-\series_{\para-1}-1$. By setting $\parb=2\series_\para$, we obtain that $|\glueset| = 8\series_\para-\series_{\para-1}-1 \leq 4\parb$. Moreover, a simple induction shows that for any $\para \geq 3$, we have $\series_{\para}\leq 3^k$. Thus, if $\para\geq 3$ then $\para \geq \log (\parb)-1$. Now, this theorem is a corollary of corollary \ref{cor:bounded}, lemma \ref{lem:longpath} and lemma \ref{lem:alwayshere}.
\end{proof}

\section{Conclusion, perpectives and open questions.}

We have shown that there exists no $2D$ pumping lemma with a linear bound according the size of the tile set for tile assembly system at temperature $1$. This result hints that this model is more complex model than initially thought. In \cite{pumpability}, the authors show a $2D$ pumping lemma with a non linear bound. Finding the exact bound is an open question. Also, we conjecture that tile assembly systems are not able to do computation at temperature $1$. Remark that these two conjectures are different. Deducing that automaton are not able to simulate Turing machine is a direct corollary of the classical pumping lemma. This is not the case for tile assembly system at temperature $1$. Indeed, some tiles may be attached to a periodic infinite path and start doing some computation. In \cite{Doty-2011}, the authors have shown that the existence of a $2D$ pumping lemma implies that directed tile assembly system (such tile assembly system have only one terminal assembly) are not able to do computation at temperature $1$. This result requires several pages and no generalization of this result is currently known. An other improvement of our result would be to modify the tile set of this paper in order to get a directed tile assembly system. Also remark that combining \cite{pumpability} and \cite{Doty-2011} means that directed tile assembly system cannot do computation at temperature $1$.

\begin{figure}

\begin{tikzpicture}[x=0.25cm,y=0.25cm]
          \clip (-1, -1) rectangle (51.000000,69.000000);
\draw[ultra thick](0.5,0.5)--(1, 0.5);\draws{0}{0}

\path [draw, gray, very thin] (-4,-2) grid[step=0.25cm] (53,71);

\draw[ultra thick, xgreen]
             (1,0.5)--(10.500000,0.5)--(10.500000,1.5)--(9.500000,1.5)--(9.500000,1.5);
\draw[ultra thick, xorange](9.500000,1.5)--(9.500000,67.5);
\drawg{1}{0}
\drawg{2}{0}
\drawg{3}{0}
\drawg{4}{0}
\drawg{5}{0}
\drawg{6}{0}
\drawg{7}{0}
\drawg{8}{0}
\drawg{9}{0}
\drawg{10}{0}
\drawg{10}{1}\drawg{9}{1}
\draw[ultra thick, xblue](9.500000, 3.500000)--(10.000000, 3.500000);
\draw[ultra thick, xblue](9.500000, 9.500000)--(10.000000, 9.500000);
\draw[ultra thick, xblue](9.500000, 25.500000)--(10.000000, 25.500000);
\draw[ultra thick, xblue](9.500000, 67.500000)--(10.000000, 67.500000);
\drawo{9}{2}
\drawo{9}{3}
\drawo{9}{4}
\drawo{9}{5}
\drawo{9}{6}
\drawo{9}{7}
\drawo{9}{8}
\drawo{9}{9}
\drawo{9}{10}
\drawo{9}{11}
\drawo{9}{12}
\drawo{9}{13}
\drawo{9}{14}
\drawo{9}{15}
\drawo{9}{16}
\drawo{9}{17}
\drawo{9}{18}
\drawo{9}{19}
\drawo{9}{20}
\drawo{9}{21}
\drawo{9}{22}
\drawo{9}{23}
\drawo{9}{24}
\drawo{9}{25}
\drawo{9}{26}
\drawo{9}{27}
\drawo{9}{28}
\drawo{9}{29}
\drawo{9}{30}
\drawo{9}{31}
\drawo{9}{32}
\drawo{9}{33}
\drawo{9}{34}
\drawo{9}{35}
\drawo{9}{36}
\drawo{9}{37}
\drawo{9}{38}
\drawo{9}{39}
\drawo{9}{40}
\drawo{9}{41}
\drawo{9}{42}
\drawo{9}{43}
\drawo{9}{44}
\drawo{9}{45}
\drawo{9}{46}
\drawo{9}{47}
\drawo{9}{48}
\drawo{9}{49}
\drawo{9}{50}
\drawo{9}{51}
\drawo{9}{52}
\drawo{9}{53}
\drawo{9}{54}
\drawo{9}{55}
\drawo{9}{56}
\drawo{9}{57}
\drawo{9}{58}
\drawo{9}{59}
\drawo{9}{60}
\drawo{9}{61}
\drawo{9}{62}
\drawo{9}{63}
\drawo{9}{64}
\drawo{9}{65}
\drawo{9}{66}
\drawo{9}{67}
\draw[ultra thick, xblue]
          (10.000000,3.500000)--(19.500000,3.500000)--(19.500000,2.500000)--(10.500000,2.500000)--(10.500000,2.500000);
\draw[ultra thick, xred](10.500000,2.000000)--(10.500000,2.500000);
\drawb{10}{3}\drawb{10}{2}
\drawb{11}{3}\drawb{11}{2}
\drawb{12}{3}\drawb{12}{2}
\drawb{13}{3}\drawb{13}{2}
\drawb{14}{3}\drawb{14}{2}
\drawb{15}{3}\drawb{15}{2}
\drawb{16}{3}\drawb{16}{2}
\drawb{17}{3}\drawb{17}{2}
\drawb{18}{3}\drawb{18}{2}
\drawb{19}{3}\drawb{19}{2}
\draw[ultra thick, xblue]
          (10.000000,9.500000)--(19.500000,9.500000)--(19.500000,8.500000)--(10.500000,8.500000)--(10.500000,8.000000);
\draw[ultra thick, xred](10.500000,8.500000)--(10.500000,4.000000);
\drawb{10}{9}\drawb{10}{8}
\drawb{11}{9}\drawb{11}{8}
\drawb{12}{9}\drawb{12}{8}
\drawb{13}{9}\drawb{13}{8}
\drawb{14}{9}\drawb{14}{8}
\drawb{15}{9}\drawb{15}{8}
\drawb{16}{9}\drawb{16}{8}
\drawb{17}{9}\drawb{17}{8}
\drawb{18}{9}\drawb{18}{8}
\drawb{19}{9}\drawb{19}{8}
\draw[ultra thick, xgreen](10.500000, 4.500000)--(11.000000, 4.500000);
\drawr{10}{4}
\drawr{10}{5}
\drawr{10}{6}
\drawr{10}{7}
\begin{scope}[shift={(10,4)}]
\draw[ultra thick, xgreen]
             (1,0.5)--(10.500000,0.5)--(10.500000,1.5)--(9.500000,1.5);
\draw[ultra thick, xorange](9.500000,1.5)--(9.500000,4.000000);
\drawg{1}{0}
\drawg{2}{0}
\drawg{3}{0}
\drawg{4}{0}
\drawg{5}{0}
\drawg{6}{0}
\drawg{7}{0}
\drawg{8}{0}
\drawg{9}{0}
\drawg{10}{0}
\drawg{10}{1}\drawg{9}{1}
\draw[ultra thick, xblue]
          (9.500000,3.500000)--(19.500000,3.500000)--(19.500000,2.500000)--(10.500000,2.500000)--(10.500000,2.500000);
\draw[ultra thick, xred]
          (10.500000,2.000000)--(10.500000,2.500000);
\drawo{9}{2}
\drawo{9}{3}
\drawb{10}{3}\drawb{10}{2}
\drawb{11}{3}\drawb{11}{2}
\drawb{12}{3}\drawb{12}{2}
\drawb{13}{3}\drawb{13}{2}
\drawb{14}{3}\drawb{14}{2}
\drawb{15}{3}\drawb{15}{2}
\drawb{16}{3}\drawb{16}{2}
\drawb{17}{3}\drawb{17}{2}
\drawb{18}{3}\drawb{18}{2}
\drawb{19}{3}\drawb{19}{2}
\end{scope}
\draw[ultra thick, xblue]
          (10.000000,25.500000)--(19.500000,25.500000)--(19.500000,24.500000)--(10.500000,24.500000);
\draw[ultra thick, xred](10.500000,24.500000)--(10.500000,10.000000);
\drawb{10}{25}\drawb{10}{24}
\drawb{11}{25}\drawb{11}{24}
\drawb{12}{25}\drawb{12}{24}
\drawb{13}{25}\drawb{13}{24}
\drawb{14}{25}\drawb{14}{24}
\drawb{15}{25}\drawb{15}{24}
\drawb{16}{25}\drawb{16}{24}
\drawb{17}{25}\drawb{17}{24}
\drawb{18}{25}\drawb{18}{24}
\drawb{19}{25}\drawb{19}{24}
\draw[ultra thick, xgreen](10.500000, 20.500000)--(11.000000, 20.500000);
\draw[ultra thick, xgreen](10.500000, 10.500000)--(11.000000, 10.500000);
\drawr{10}{10}
\drawr{10}{11}
\drawr{10}{12}
\drawr{10}{13}
\drawr{10}{14}
\drawr{10}{15}
\drawr{10}{16}
\drawr{10}{17}
\drawr{10}{18}
\drawr{10}{19}
\drawr{10}{20}
\drawr{10}{21}
\drawr{10}{22}
\drawr{10}{23}
\begin{scope}[shift={(10,20)}]
\draw[ultra thick, xgreen]
             (1,0.5)--(10.500000,0.5)--(10.500000,1.5)--(9.500000,1.5)--(9.500000,1.5);
\draw[ultra thick, xorange](9.500000,1.5)--(9.500000,4.000000);
\draw[ultra thick, xblue](9.500000, 3.500000)--(10.000000, 3.500000);
\drawg{1}{0}
\drawg{2}{0}
\drawg{3}{0}
\drawg{4}{0}
\drawg{5}{0}
\drawg{6}{0}
\drawg{7}{0}
\drawg{8}{0}
\drawg{9}{0}
\drawg{10}{0}
\drawg{10}{1}\drawg{9}{1}
\drawo{9}{2}
\drawo{9}{3}
\draw[ultra thick, xblue]
          (10.000000,3.500000)--(19.500000,3.500000)--(19.500000,2.500000)--(10.500000,2.500000)--(10.500000,2.000000);
\draw[ultra thick, xred](10.500000,2.500000)--(10.500000,2.000000);
\drawb{10}{3}\drawb{10}{2}
\drawb{11}{3}\drawb{11}{2}
\drawb{12}{3}\drawb{12}{2}
\drawb{13}{3}\drawb{13}{2}
\drawb{14}{3}\drawb{14}{2}
\drawb{15}{3}\drawb{15}{2}
\drawb{16}{3}\drawb{16}{2}
\drawb{17}{3}\drawb{17}{2}
\drawb{18}{3}\drawb{18}{2}
\drawb{19}{3}\drawb{19}{2}
\end{scope}
\begin{scope}[shift={(10,10)}]
\draw[ultra thick, xgreen]
             (1,0.5)--(10.500000,0.5)--(10.500000,1.5)--(9.500000,1.5)--(9.500000,1.5);
\draw[ultra thick, xorange](9.500000,1.5)--(9.500000,10.000000);
\drawg{1}{0}
\drawg{2}{0}
\drawg{3}{0}
\drawg{4}{0}
\drawg{5}{0}
\drawg{6}{0}
\drawg{7}{0}
\drawg{8}{0}
\drawg{9}{0}
\drawg{10}{0}
\drawg{10}{1}\drawg{9}{1}
\draw[ultra thick, xblue](9.500000, 3.500000)--(10.000000, 3.500000);
\draw[ultra thick, xblue](9.500000, 9.500000)--(10.000000, 9.500000);
\drawo{9}{2}
\drawo{9}{3}
\drawo{9}{4}
\drawo{9}{5}
\drawo{9}{6}
\drawo{9}{7}
\drawo{9}{8}
\drawo{9}{9}
\draw[ultra thick, xblue]
          (10.000000,3.500000)--(19.500000,3.500000)--(19.500000,2.500000)--(10.500000,2.500000);
\draw[ultra thick, xred](10.500000,2.500000)--(10.500000,2.000000);
\drawb{10}{3}\drawb{10}{2}
\drawb{11}{3}\drawb{11}{2}
\drawb{12}{3}\drawb{12}{2}
\drawb{13}{3}\drawb{13}{2}
\drawb{14}{3}\drawb{14}{2}
\drawb{15}{3}\drawb{15}{2}
\drawb{16}{3}\drawb{16}{2}
\drawb{17}{3}\drawb{17}{2}
\drawb{18}{3}\drawb{18}{2}
\drawb{19}{3}\drawb{19}{2}
\draw[ultra thick, xblue]
          (10.000000,9.500000)--(19.500000,9.500000)--(19.500000,8.500000)--(10.500000,8.500000)--(10.500000,8.500000);
\draw[ultra thick, xred](10.500000,8.500000)--(10.500000,4.000000);
\drawb{10}{9}\drawb{10}{8}
\drawb{11}{9}\drawb{11}{8}
\drawb{12}{9}\drawb{12}{8}
\drawb{13}{9}\drawb{13}{8}
\drawb{14}{9}\drawb{14}{8}
\drawb{15}{9}\drawb{15}{8}
\drawb{16}{9}\drawb{16}{8}
\drawb{17}{9}\drawb{17}{8}
\drawb{18}{9}\drawb{18}{8}
\drawb{19}{9}\drawb{19}{8}
\draw[ultra thick, xgreen](10.500000, 4.500000)--(11.000000, 4.500000);
\drawr{10}{4}
\drawr{10}{5}
\drawr{10}{6}
\drawr{10}{7}
\begin{scope}[shift={(10,4)}]
\draw[ultra thick, xgreen]
             (1,0.5)--(10.500000,0.5)--(10.500000,1.5)--(9.500000,1.5);
\draw[ultra thick, xorange](9.500000,1.5)--(9.500000,4.000000);
\drawg{1}{0}
\drawg{2}{0}
\drawg{3}{0}
\drawg{4}{0}
\drawg{5}{0}
\drawg{6}{0}
\drawg{7}{0}
\drawg{8}{0}
\drawg{9}{0}
\drawg{10}{0}
\drawg{10}{1}\drawg{9}{1}
\draw[ultra thick, xblue](9.500000, 3.500000)--(10.000000, 3.500000);
\drawo{9}{2}
\drawo{9}{3}
\draw[ultra thick, xblue]
          (10.000000,3.500000)--(19.500000,3.500000)--(19.500000,2.500000)--(10.500000,2.500000)--(10.500000,2.500000);
\draw[ultra thick, xred](10.500000,2.000000)--(10.500000,2.500000);
\drawb{10}{3}\drawb{10}{2}
\drawb{11}{3}\drawb{11}{2}
\drawb{12}{3}\drawb{12}{2}
\drawb{13}{3}\drawb{13}{2}
\drawb{14}{3}\drawb{14}{2}
\drawb{15}{3}\drawb{15}{2}
\drawb{16}{3}\drawb{16}{2}
\drawb{17}{3}\drawb{17}{2}
\drawb{18}{3}\drawb{18}{2}
\drawb{19}{3}\drawb{19}{2}
\end{scope}
\end{scope}
\draw[ultra thick, xblue]
          (10.000000,67.500000)--(19.500000,67.500000)--(19.500000,66.500000)--(10.500000,66.500000)--(10.500000,66.500000);
\draw[ultra thick, xred](10.500000,66.500000)--(10.500000,26.500000);
\drawb{10}{67}\drawb{10}{66}
\drawb{11}{67}\drawb{11}{66}
\drawb{12}{67}\drawb{12}{66}
\drawb{13}{67}\drawb{13}{66}
\drawb{14}{67}\drawb{14}{66}
\drawb{15}{67}\drawb{15}{66}
\drawb{16}{67}\drawb{16}{66}
\drawb{17}{67}\drawb{17}{66}
\drawb{18}{67}\drawb{18}{66}
\drawb{19}{67}\drawb{19}{66}
\draw[ultra thick, xgreen](10.500000, 62.500000)--(11.000000, 62.500000);
\draw[ultra thick, xgreen](10.500000, 52.500000)--(11.000000, 52.500000);
\draw[ultra thick, xgreen](10.500000, 26.500000)--(11.000000, 26.500000);
\drawr{10}{26}
\drawr{10}{27}
\drawr{10}{28}
\drawr{10}{29}
\drawr{10}{30}
\drawr{10}{31}
\drawr{10}{32}
\drawr{10}{33}
\drawr{10}{34}
\drawr{10}{35}
\drawr{10}{36}
\drawr{10}{37}
\drawr{10}{38}
\drawr{10}{39}
\drawr{10}{40}
\drawr{10}{41}
\drawr{10}{42}
\drawr{10}{43}
\drawr{10}{44}
\drawr{10}{45}
\drawr{10}{46}
\drawr{10}{47}
\drawr{10}{48}
\drawr{10}{49}
\drawr{10}{50}
\drawr{10}{51}
\drawr{10}{52}
\drawr{10}{53}
\drawr{10}{54}
\drawr{10}{55}
\drawr{10}{56}
\drawr{10}{57}
\drawr{10}{58}
\drawr{10}{59}
\drawr{10}{60}
\drawr{10}{61}
\drawr{10}{62}
\drawr{10}{63}
\drawr{10}{64}
\drawr{10}{65}
\begin{scope}[shift={(10,62)}]
\draw[ultra thick, xgreen]
             (1,0.5)--(10.500000,0.5)--(10.500000,1.5)--(9.500000,1.5)--(9.500000,1.5);
\draw[ultra thick, xorange](9.500000,1.5)--(9.500000,4.000000);
\drawg{1}{0}
\drawg{2}{0}
\drawg{3}{0}
\drawg{4}{0}
\drawg{5}{0}
\drawg{6}{0}
\drawg{7}{0}
\drawg{8}{0}
\drawg{9}{0}
\drawg{10}{0}
\drawg{10}{1}\drawg{9}{1}
\draw[ultra thick, xblue](9.500000, 3.500000)--(10.000000, 3.500000);
\drawo{9}{2}
\drawo{9}{3}
\draw[ultra thick, xblue]
          (10.000000,3.500000)--(19.500000,3.500000)--(19.500000,2.500000)--(10.500000,2.500000)--(10.500000,2.500000);
\draw[ultra thick, xred](10.500000,2.500000)--(10.500000,2.000000);
\drawb{10}{3}\drawb{10}{2}
\drawb{11}{3}\drawb{11}{2}
\drawb{12}{3}\drawb{12}{2}
\drawb{13}{3}\drawb{13}{2}
\drawb{14}{3}\drawb{14}{2}
\drawb{15}{3}\drawb{15}{2}
\drawb{16}{3}\drawb{16}{2}
\drawb{17}{3}\drawb{17}{2}
\drawb{18}{3}\drawb{18}{2}
\drawb{19}{3}\drawb{19}{2}
\end{scope}
\begin{scope}[shift={(10,52)}]
\draw[ultra thick, xgreen]
             (1,0.5)--(10.500000,0.5)--(10.500000,1.5)--(9.500000,1.5)--(9.500000,1.5);
\draw[ultra thick, xorange](9.500000,1.5)--(9.500000,10.000000);
\drawg{1}{0}
\drawg{2}{0}
\drawg{3}{0}
\drawg{4}{0}
\drawg{5}{0}
\drawg{6}{0}
\drawg{7}{0}
\drawg{8}{0}
\drawg{9}{0}
\drawg{10}{0}
\drawg{10}{1}\drawg{9}{1}
\draw[ultra thick, xblue](9.500000, 3.500000)--(10.000000, 3.500000);
\draw[ultra thick, xblue](9.500000, 9.500000)--(10.000000, 9.500000);
\drawo{9}{2}
\drawo{9}{3}
\drawo{9}{4}
\drawo{9}{5}
\drawo{9}{6}
\drawo{9}{7}
\drawo{9}{8}
\drawo{9}{9}
\draw[ultra thick, xblue]
          (10.000000,3.500000)--(19.500000,3.500000)--(19.500000,2.500000)--(10.500000,2.500000);
\draw[ultra thick, xred](10.500000,2.500000)--(10.500000,2.000000);
\drawb{10}{3}\drawb{10}{2}
\drawb{11}{3}\drawb{11}{2}
\drawb{12}{3}\drawb{12}{2}
\drawb{13}{3}\drawb{13}{2}
\drawb{14}{3}\drawb{14}{2}
\drawb{15}{3}\drawb{15}{2}
\drawb{16}{3}\drawb{16}{2}
\drawb{17}{3}\drawb{17}{2}
\drawb{18}{3}\drawb{18}{2}
\drawb{19}{3}\drawb{19}{2}
\draw[ultra thick, xblue]
          (10.000000,9.500000)--(19.500000,9.500000)--(19.500000,8.500000)--(10.500000,8.500000)--(10.500000,8.500000);
\draw[ultra thick, xred](10.500000,8.500000)--(10.500000,4.000000);
\drawb{10}{9}\drawb{10}{8}
\drawb{11}{9}\drawb{11}{8}
\drawb{12}{9}\drawb{12}{8}
\drawb{13}{9}\drawb{13}{8}
\drawb{14}{9}\drawb{14}{8}
\drawb{15}{9}\drawb{15}{8}
\drawb{16}{9}\drawb{16}{8}
\drawb{17}{9}\drawb{17}{8}
\drawb{18}{9}\drawb{18}{8}
\drawb{19}{9}\drawb{19}{8}
\draw[ultra thick, xgreen](10.500000, 4.500000)--(11.000000, 4.500000);
\drawr{10}{4}
\drawr{10}{5}
\drawr{10}{6}
\drawr{10}{7}
\begin{scope}[shift={(10,4)}]
\draw[ultra thick, xgreen]
             (1,0.5)--(10.500000,0.5)--(10.500000,1.5)--(9.500000,1.5)--(9.500000,1.5);
\draw[ultra thick, xorange](9.500000,1.5)--(9.500000,4.000000);
\drawg{1}{0}
\drawg{2}{0}
\drawg{3}{0}
\drawg{4}{0}
\drawg{5}{0}
\drawg{6}{0}
\drawg{7}{0}
\drawg{8}{0}
\drawg{9}{0}
\drawg{10}{0}
\drawg{10}{1}\drawg{9}{1}
\draw[ultra thick, xblue](9.500000, 3.500000)--(10.000000, 3.500000);
\drawo{9}{2}
\drawo{9}{3}
\draw[ultra thick, xblue]
          (10.000000,3.500000)--(19.500000,3.500000)--(19.500000,2.500000)--(10.500000,2.500000)--(10.500000,2.500000);
\draw[ultra thick, xred](10.500000,2.500000)--(10.500000,2.000000);
\drawb{10}{3}\drawb{10}{2}
\drawb{11}{3}\drawb{11}{2}
\drawb{12}{3}\drawb{12}{2}
\drawb{13}{3}\drawb{13}{2}
\drawb{14}{3}\drawb{14}{2}
\drawb{15}{3}\drawb{15}{2}
\drawb{16}{3}\drawb{16}{2}
\drawb{17}{3}\drawb{17}{2}
\drawb{18}{3}\drawb{18}{2}
\drawb{19}{3}\drawb{19}{2}
\end{scope}
\end{scope}
\begin{scope}[shift={(10,26)}]
\draw[ultra thick, xgreen]
             (1,0.5)--(10.500000,0.5)--(10.500000,1.5)--(9.500000,1.5)--(9.500000,2);
\draw[ultra thick, xorange](9.500000,1.5)--(9.500000,26.000000);
\drawg{1}{0}
\drawg{2}{0}
\drawg{3}{0}
\drawg{4}{0}
\drawg{5}{0}
\drawg{6}{0}
\drawg{7}{0}
\drawg{8}{0}
\drawg{9}{0}
\drawg{10}{0}
\drawg{10}{1}\drawg{9}{1}
\draw[ultra thick, xblue](9.500000, 3.500000)--(10.000000, 3.500000);
\draw[ultra thick, xblue](9.500000, 9.500000)--(10.000000, 9.500000);
\draw[ultra thick, xblue](9.500000, 25.500000)--(10.000000, 25.500000);
\drawo{9}{2}
\drawo{9}{3}
\drawo{9}{4}
\drawo{9}{5}
\drawo{9}{6}
\drawo{9}{7}
\drawo{9}{8}
\drawo{9}{9}
\drawo{9}{10}
\drawo{9}{11}
\drawo{9}{12}
\drawo{9}{13}
\drawo{9}{14}
\drawo{9}{15}
\drawo{9}{16}
\drawo{9}{17}
\drawo{9}{18}
\drawo{9}{19}
\drawo{9}{20}
\drawo{9}{21}
\drawo{9}{22}
\drawo{9}{23}
\drawo{9}{24}
\drawo{9}{25}
\draw[ultra thick, xblue]
          (10.000000,3.500000)--(19.500000,3.500000)--(19.500000,2.500000)--(10.500000,2.500000)--(10.500000,2.500000);
\draw[ultra thick, xred](10.500000,2.000000)--(10.500000,2.500000);
\drawb{10}{3}\drawb{10}{2}
\drawb{11}{3}\drawb{11}{2}
\drawb{12}{3}\drawb{12}{2}
\drawb{13}{3}\drawb{13}{2}
\drawb{14}{3}\drawb{14}{2}
\drawb{15}{3}\drawb{15}{2}
\drawb{16}{3}\drawb{16}{2}
\drawb{17}{3}\drawb{17}{2}
\drawb{18}{3}\drawb{18}{2}
\drawb{19}{3}\drawb{19}{2}
\draw[ultra thick, xblue]
          (10.000000,9.500000)--(19.500000,9.500000)--(19.500000,8.500000)--(10.500000,8.500000)--(10.500000,8.500000);
\draw[ultra thick, xred](10.500000,8.500000)--(10.500000,4.000000);
\drawb{10}{9}\drawb{10}{8}
\drawb{11}{9}\drawb{11}{8}
\drawb{12}{9}\drawb{12}{8}
\drawb{13}{9}\drawb{13}{8}
\drawb{14}{9}\drawb{14}{8}
\drawb{15}{9}\drawb{15}{8}
\drawb{16}{9}\drawb{16}{8}
\drawb{17}{9}\drawb{17}{8}
\drawb{18}{9}\drawb{18}{8}
\drawb{19}{9}\drawb{19}{8}
\draw[ultra thick, xred](10.500000, 4.500000)--(11.000000, 4.500000);
\drawr{10}{4}
\drawr{10}{5}
\drawr{10}{6}
\drawr{10}{7}
\begin{scope}[shift={(10,4)}]
\draw[ultra thick, xgreen]
             (1,0.5)--(10.500000,0.5)--(10.500000,1.5)--(9.500000,1.5)--(9.500000,2);
\drawg{1}{0}
\drawg{2}{0}
\drawg{3}{0}
\drawg{4}{0}
\drawg{5}{0}
\drawg{6}{0}
\drawg{7}{0}
\drawg{8}{0}
\drawg{9}{0}
\drawg{10}{0}
\drawg{10}{1}\drawg{9}{1}
\draw[ultra thick, xorange](9.500000,2)--(9.500000,4.000000);
\draw[ultra thick, xorange](9.500000, 3.500000)--(10.000000, 3.500000);
\drawo{9}{2}
\drawo{9}{3}
\draw[ultra thick, xblue]
          (10.000000,3.500000)--(19.500000,3.500000)--(19.500000,2.500000)--(10.500000,2.500000)--(10.500000,2.500000);
\draw[ultra thick, xred](10.500000,2.000000)--(10.500000,2.500000);
\drawb{10}{3}\drawb{10}{2}
\drawb{11}{3}\drawb{11}{2}
\drawb{12}{3}\drawb{12}{2}
\drawb{13}{3}\drawb{13}{2}
\drawb{14}{3}\drawb{14}{2}
\drawb{15}{3}\drawb{15}{2}
\drawb{16}{3}\drawb{16}{2}
\drawb{17}{3}\drawb{17}{2}
\drawb{18}{3}\drawb{18}{2}
\drawb{19}{3}\drawb{19}{2}
\end{scope}
\draw[ultra thick, xblue]
          (10.000000,25.500000)--(19.500000,25.500000)--(19.500000,24.500000)--(10.500000,24.500000)--(10.500000,24.000000);
\drawb{10}{25}\drawb{10}{24}
\drawb{11}{25}\drawb{11}{24}
\drawb{12}{25}\drawb{12}{24}
\drawb{13}{25}\drawb{13}{24}
\drawb{14}{25}\drawb{14}{24}
\drawb{15}{25}\drawb{15}{24}
\drawb{16}{25}\drawb{16}{24}
\drawb{17}{25}\drawb{17}{24}
\drawb{18}{25}\drawb{18}{24}
\drawb{19}{25}\drawb{19}{24}
\draw[ultra thick, xred](10.500000,24.000000)--(10.500000,10.000000);
\draw[ultra thick, xred](10.500000, 20.500000)--(11.000000, 20.500000);
\draw[ultra thick, xred](10.500000, 10.500000)--(11.000000, 10.500000);
\drawr{10}{10}
\drawr{10}{11}
\drawr{10}{12}
\drawr{10}{13}
\drawr{10}{14}
\drawr{10}{15}
\drawr{10}{16}
\drawr{10}{17}
\drawr{10}{18}
\drawr{10}{19}
\drawr{10}{20}
\drawr{10}{21}
\drawr{10}{22}
\drawr{10}{23}
\begin{scope}[shift={(10,20)}]
\draw[ultra thick, xgreen]
             (1,0.5)--(10.500000,0.5)--(10.500000,1.5)--(9.500000,1.5)--(9.500000,1.5);
\draw[ultra thick, xorange](9.500000,1.5)--(9.500000,4.000000);
\drawg{1}{0}
\drawg{2}{0}
\drawg{3}{0}
\drawg{4}{0}
\drawg{5}{0}
\drawg{6}{0}
\drawg{7}{0}
\drawg{8}{0}
\drawg{9}{0}
\drawg{10}{0}
\drawg{10}{1}\drawg{9}{1}
\draw[ultra thick, xblue](9.500000, 3.500000)--(10.000000, 3.500000);
\drawo{9}{2}
\drawo{9}{3}
\draw[ultra thick, xblue]
          (10.000000,3.500000)--(19.500000,3.500000)--(19.500000,2.500000)--(10.500000,2.500000)--(10.500000,2.500000);
\draw[ultra thick, xred](10.500000,2.500000)--(10.500000,2.000000);
\drawb{10}{3}\drawb{10}{2}
\drawb{11}{3}\drawb{11}{2}
\drawb{12}{3}\drawb{12}{2}
\drawb{13}{3}\drawb{13}{2}
\drawb{14}{3}\drawb{14}{2}
\drawb{15}{3}\drawb{15}{2}
\drawb{16}{3}\drawb{16}{2}
\drawb{17}{3}\drawb{17}{2}
\drawb{18}{3}\drawb{18}{2}
\drawb{19}{3}\drawb{19}{2}
\end{scope}
\begin{scope}[shift={(10,10)}]
\draw[ultra thick, xgreen]
             (1,0.5)--(10.500000,0.5)--(10.500000,1.5)--(9.500000,1.5)--(9.500000,1.5);
\draw[ultra thick, xorange](9.500000,1.5)--(9.500000,10.000000);
\drawg{1}{0}
\drawg{2}{0}
\drawg{3}{0}
\drawg{4}{0}
\drawg{5}{0}
\drawg{6}{0}
\drawg{7}{0}
\drawg{8}{0}
\drawg{9}{0}
\drawg{10}{0}
\drawg{10}{1}\drawg{9}{1}
\draw[ultra thick, xblue](9.500000, 3.500000)--(10.000000, 3.500000);
\draw[ultra thick, xblue](9.500000, 9.500000)--(10.000000, 9.500000);
\drawo{9}{2}
\drawo{9}{3}
\drawo{9}{4}
\drawo{9}{5}
\drawo{9}{6}
\drawo{9}{7}
\drawo{9}{8}
\drawo{9}{9}
\draw[ultra thick, xblue]
          (10.000000,3.500000)--(19.500000,3.500000)--(19.500000,2.500000)--(10.500000,2.500000);
\draw[ultra thick, xred](10.500000,2.500000)--(10.500000,2.000000);
\drawb{10}{3}\drawb{10}{2}
\drawb{11}{3}\drawb{11}{2}
\drawb{12}{3}\drawb{12}{2}
\drawb{13}{3}\drawb{13}{2}
\drawb{14}{3}\drawb{14}{2}
\drawb{15}{3}\drawb{15}{2}
\drawb{16}{3}\drawb{16}{2}
\drawb{17}{3}\drawb{17}{2}
\drawb{18}{3}\drawb{18}{2}
\drawb{19}{3}\drawb{19}{2}
\draw[ultra thick, xblue]
          (10.000000,9.500000)--(19.500000,9.500000)--(19.500000,8.500000)--(10.500000,8.500000)--(10.500000,8.500000);
\draw[ultra thick, xred](10.500000,8.500000)--(10.500000,4.000000);
\drawb{10}{9}\drawb{10}{8}
\drawb{11}{9}\drawb{11}{8}
\drawb{12}{9}\drawb{12}{8}
\drawb{13}{9}\drawb{13}{8}
\drawb{14}{9}\drawb{14}{8}
\drawb{15}{9}\drawb{15}{8}
\drawb{16}{9}\drawb{16}{8}
\drawb{17}{9}\drawb{17}{8}
\drawb{18}{9}\drawb{18}{8}
\drawb{19}{9}\drawb{19}{8}
\draw[ultra thick, xgreen](10.500000, 4.500000)--(11.000000, 4.500000);
\drawr{10}{4}
\drawr{10}{5}
\drawr{10}{6}
\drawr{10}{7}
\begin{scope}[shift={(10,4)}]
\draw[ultra thick, xgreen]
             (1,0.5)--(10.500000,0.5)--(10.500000,1.5)--(9.500000,1.5)--(9.500000,1.5);
\draw[ultra thick, xorange](9.500000,1.5)--(9.500000,4.000000);
\drawg{1}{0}
\drawg{2}{0}
\drawg{3}{0}
\drawg{4}{0}
\drawg{5}{0}
\drawg{6}{0}
\drawg{7}{0}
\drawg{8}{0}
\drawg{9}{0}
\drawg{10}{0}
\drawg{10}{1}\drawg{9}{1}
\draw[ultra thick, xblue](9.500000, 3.500000)--(10.000000, 3.500000);
\drawo{9}{2}
\drawo{9}{3}
\draw[ultra thick, xblue]
          (10.000000,3.500000)--(19.500000,3.500000)--(19.500000,2.500000)--(10.500000,2.500000)--(10.500000,2.500000);
\draw[ultra thick, xred](10.500000,2.500000)--(10.500000,2.000000);
\drawb{10}{3}\drawb{10}{2}
\drawb{11}{3}\drawb{11}{2}
\drawb{12}{3}\drawb{12}{2}
\drawb{13}{3}\drawb{13}{2}
\drawb{14}{3}\drawb{14}{2}
\drawb{15}{3}\drawb{15}{2}
\drawb{16}{3}\drawb{16}{2}
\drawb{17}{3}\drawb{17}{2}
\drawb{18}{3}\drawb{18}{2}
\drawb{19}{3}\drawb{19}{2}
\end{scope}
\end{scope}
\end{scope}
\end{tikzpicture}

\caption{A terminal assembly of $\tas$ for $\para=4$ and $\parb=10$.}
\label{fig:terminal}
\end{figure}

\section*{Acknowledgments}

We would like to thank Damien Woods for invaluable discussions, comments, and suggestions.

\bibliographystyle{plain}
\bibliography{t1_tight_bound}

%
%
%
%
%
%
\end{document}